\documentclass{article}
\usepackage{amsmath, amsthm, amssymb,  mathrsfs}
\usepackage{graphicx}
\usepackage{float}
\title{Kastor-Traschen Black Holes, Null Geodesics and Conformal Circles}
\author{Stephen Casey\footnote{\textbf{Email:} sc581@cam.ac.uk} \\Department of Applied Mathematics and Theoretical Physics, \\University of Cambridge, \\Wilberforce Road, Cambridge CB3 0WA, UK.}

\newtheorem{twocent}{Proposition}[section]
\newtheorem{twocenttwo}[twocent]{Proposition}
\newtheorem{charac}[twocent]{Proposition}
\newtheorem{finalprop}[twocent]{Proposition}

\begin{document}
\maketitle
\begin{abstract}
The Kastor-Traschen metric is a time-dependent solution of the Einstein-Maxwell equations with positive cosmological constant $\Lambda$ which can be used to describe an arbitrary number of charged dynamical black holes. In this paper, we consider the null geodesic structure of this solution, in particular, focusing on the projection to the space of orbits of the timelike conformal retraction. It is found that these projected light rays arise as integral curves of a system of third order ordinary differential equations. This system is not uniquely defined, however, and we use the inherent freedom to construct a new system whose integral curves coincide with the projection of distinguished null curves of Kastor-Traschen arising from a magnetic flow. We discuss our results in the one-centre case and demonstrate a link to conformal circles in the limit $\Lambda \rightarrow 0$. We also show how to construct analytic expressions for the projected null geodesics of this metric by exploiting a well-known diffeomorphism between the K-T metric and extremal Reissner-Nordstrom deSitter. We make some remarks about the two-centre solution and demonstrate a link with the one-centre case.

\end{abstract}

\section{Introduction}
Recently, the study of the behaviour of null geodesics has led to the development of geometric tools useful for revealing physical properties of spacetimes which admit a particular type of timelike vector field.
\\In the case of a \emph{static metric}, for example, a convenient geometric tool is the \emph{optical metric} defined on the space of orbits of a given hypersurface-orthogonal (HSO) Killing vector. Light rays of the static metric are found to project down to unparametrised geodesics of the associated optical metric. This property has been used to understand some observations of the dynamics of light rays in Schwarzschild-deSitter metrics \cite{wernerwar}, and to give an alternative interpretation of black-hole no-hair theorems \cite{blackhole}. Static metrics admitting more than one HSO Killing vector have been considered in the context of projective equivalence of the associated optical metrics \cite{self}.
\\Some progress has also been made in the case of \emph{stationary metrics} \cite{statmet}, where there are two distinguished geometric structures on the space of orbits $\mathcal{B}$ of the timelike Killing vector field. On one hand, the light rays of the stationary metric project down to solutions of the Zermelo navigation problem on $\mathcal{B}$ with specially-defined metric $h$ and wind $\textbf{W}$. Alternatively, one can construct a Randers-Finsler structure on $\mathcal{B}$ with Randers data $\{ a_{ij},b_{i} \}$ such that the null geodesic flow projects down to the magnetic flow due to $db$ on the curved manifold $\{\mathcal{B}, a_{ij} \}$.
\\In this paper, we generalise this work to another class of metrics - those which admit a timelike \emph{conformal retraction}, i.e, there exists a hypersurface-orthogonal timelike vector field, $\Theta$, for which the conformal structure on its space of orbits is preserved along the integral curves of $\Theta$. In the Riemannian case, such metrics have arisen as supersymmetric solutions of minimal $\mathcal{N}=2$ gauged supergravity with anti-self-dual Maxwell field \cite{DGST}. Moreover, when the anti-self-duality condition is relaxed in the case of positive cosmological constant, one obtains a solution, also admitting a conformal retraction, which is the Riemannian analogue of the well-known \emph{Kastor-Traschen} metric \cite{kastor}.
\\In Lorentzian signature, this metric is a time-dependent solution to the Einstein-Maxwell equations which can be seen to describe an arbitrary number of dynamical charged black holes in a deSitter background. In this paper, we focus on the null geodesic structure of these metrics and in particular, it is found that light rays project down to integral curves of a system of third order ordinary differential equations (ODEs) on the space of orbits of $\Theta$. In particular, we find that the projected null geodesics form only a subset of the set of integral curves and that this yields a freedom in how this third order system is defined.
\\In section 4, we pay particular attention to the one-centre solution (single black hole). We show that, in the limit as the cosmological constant tends to zero, our system of ODEs becomes that describing conformal circles of the flat metric (as described in \cite{yanocc}). We use this result to motivate section 5 where we discuss the numerology of the problem and construct a new third order system in three dimensions. The advantage of this new system is that it will allow us to give a physical interpretation to those integral curves which do not arise as projected null geodesics, coinciding with motion in the background magnetic field. This formulation also allows us to characterize the projected light rays. Proposition 5.1 is the central result of this paper.
\\In section 6, we show how to derive analytic expressions for projected null geodesics of the one-centre case by making use of a diffeomorphism between the extreme Reissner-Nordstrom deSitter metric and the Kastor-Traschen metric. We give plots of some of these curves and discuss the horizon structure in both sets of coordinates.
\\In section 7, we discuss the Kastor-Traschen solution with two centres. We demonstrate that, for special initial conditions, the projection along the conformal retraction of a null geodesic will lie in a plane. For one such curve, we illustrate a connection between the null geodesics of the two-centre Kastor-Traschen metric and a third order system that arises in the analysis of the one-centre case. We also look at the perturbations away from this plane and give a strict condition for stability.
\\Throughout this paper, we refer to the projection of geodesics onto the space of orbits of the conformal retraction $\Theta$ as the \emph{retraction projection} to avoid ambiguity in later sections.

\renewcommand{\abstractname}{Acknowledgements}
\begin{abstract}
I would like to thank my supervisor Maciej Dunajski for suggesting the idea for this paper and for his invaluable input into the development of the work. I would also like to thank Paul Tod and Christian Luebbe for helpful discussions.
\end{abstract}

\section{Conformal Retraction in the Kastor-Traschen metric}
The Kastor-Traschen metrics are a class of time dependent solutions to the Einstein-Maxwell equations with positive cosmological constant $\Lambda$ \cite{kastor}. In local coordinates, these metrics may be written as
\begin{equation}
g = -\frac{dT^2}{(V+cT)^2} + (V+cT)^2 h
\label{firsteq}
\end{equation}
with Maxwell 1-form
\begin{equation*}
A = \frac{dT}{V+cT},
\end{equation*}
where $V = V(x)$ is a harmonic function on the spatial coordinates, $c = \pm \sqrt{\frac{\Lambda}{3}}$ is a constant and
\begin{equation*}
h = h_{ij} dx^i dx^j \,\,\,\,\,,\,\,\,\,\, i,j = 1,2,3
\end{equation*}
is the flat 3-dimensional Riemannian metric (which we express in Cartesian coordinates for now).
\\Notice here that since $V$ is a harmonic function there is a freedom in defining the electromagnetic field tensor $F = dA$ so that the Einstein-Maxwell equations are still satisfied i.e, we can write
\begin{equation}
F = \frac{1}{\sqrt{1+\nu^2}} \left( dA -  \nu \epsilon_{ijk} h^{il} \frac{\partial V}{\partial x^l} dx^j \wedge dx^k \right) \,\,\,\,\,,\,\,\,\,\, i,j,k = 1,2,3
\label{newf}
\end{equation}
where $\nu$ is a constant and $\epsilon_{ijk}$ is totally anti-symmetric in its indices with $\epsilon_{123} = 1$. This allows us to introduce a magnetic field $\textbf{B} \propto \nabla V$ into our definition of the Kastor-Traschen solution, a notion which will be useful later.
\\In the limit as $c \rightarrow 0$, these metrics reduce to the well-known Majumdar-Papapetrou solutions \cite{majumdar}, \cite{papapetrou}. For the M-P metrics, it can be shown that solutions with
\begin{equation*}
V = \displaystyle \sum_{\alpha=1}^N \frac{m_\alpha}{|\textbf{x} - \textbf{w}_\alpha|},
\end{equation*}
where $|\textbf{x} - \textbf{w}_\alpha| = (h_{jk} (x^j-w_\alpha^j)(x^k - w_\alpha^k))^{1/2}$ and $\textbf{w}_{\alpha}$ is a fixed vector for each $\alpha$, can be analytically extended to be interpreted as a system of charge equal mass black holes \cite{hartle}. In this system, the gravitational forces on each black hole are balanced by the electrostatic forces. However, when $c \neq 0$, the black holes are dynamic and can be observed to coalesce \cite{kastor}.
\\As noted in the introduction, an interesting property of the K-T metrics is that they admit a timelike \emph{conformal retraction} $\Theta$. In other words,
\begin{equation*}
\mathcal{L}_{\Theta} H_{\mu \nu} = f H_{\mu \nu} + \Theta_{(\mu}C_{\nu)}
\end{equation*}
where
\begin{equation*}
H_{\mu \nu} = g_{\mu \nu} - \frac{\Theta_\mu \Theta_\nu}{¦\Theta¦^2} \,\,\,\,\,,\,\,\,\,\, \Theta^2 := g_{\mu \nu} \Theta^{\mu} \Theta^{\nu}
\end{equation*}
and $f$ and $C$ are an arbitrary function and one-form, respectively. Greek indices here run over the values 0,1,2,3 and we raise and lower indices using the metric $g_{\mu \nu}$. In our coordinates (\ref{firsteq}), the conformal retraction is $\Theta = \frac{\partial}{\partial T}$. Furthermore, the tensor $H$ is given by
\begin{equation*}
H_{0\mu} = 0 \,\,\,\,\,,\,\,\,\,\, H_{\mu \nu} = g_{\mu \nu} \,\,\,\,\,  \text{otherwise}.
\end{equation*}
The Lie derivative of this tensor is easy to compute
\begin{equation*}
(\mathcal{L}_{\Theta} H)_{\mu \nu} = \Theta^{\lambda} H_{\mu \nu, \lambda} + H_{\lambda \nu} \Theta^{\lambda}_{,\mu} + H_{\mu \lambda} \Theta^{\lambda}_{,\nu} = \frac{\partial}{\partial T} H_{\mu \nu} = 2c (V+cT) H_{\mu \nu}.
\end{equation*}
Hence, in this case, our function $f$ and one-form $C$ are given by
\begin{equation*}
f = 2c(V+cT) \,\,\,\,\,,\,\,\,\,\, C = 0.
\end{equation*}
Now, let us decree that under a change of metric $\tilde{g} = \Omega^2 g$, the choice of conformal retraction remains unchanged i.e, $\tilde{\Theta} = \Theta = \frac{\partial}{\partial T}$. Then, $\tilde{H}_{\mu \nu} = \Omega^2 H_{\mu \nu}$ and
\begin{equation*}
(\mathcal{L}_{\tilde{\Theta}} \tilde{H})_{\mu \nu} = (\mathcal{L}_{\Theta} \Omega^2 H)_{\mu \nu} = (\mathcal{L}_{\Theta} \Omega^2) H_{\mu \nu} + \Omega^2 (\mathcal{L}_{\Theta} H)_{\mu \nu} = (f + 2\mathcal{L}_{\Theta} \log \Omega)\tilde{H}_{\mu \nu} + \Omega^2 \Theta_{(\mu} C_{\nu)}.
\end{equation*}
Hence,
\begin{equation*}
\tilde{f} = f + 2\mathcal{L}_{\Theta} \log \Omega \,\,\,\,\,,\,\,\,\,\, \tilde{C} = \Omega^2 C.
\end{equation*}
and, for our example, $C = 0$ for any choice of metric in the conformal class of $g$.

\section{Projection of Null geodesics with arc-length parametrisation}
Since we are considering only the null geodesic structure of the Kastor-Traschen metrics, we may as well begin with a conformally rescaled version of (\ref{firsteq}) - this will ease the computation a little. So, let us take a new definition of $g$
\begin{equation*}
g \rightarrow \frac{1}{(V+cT)^2} g =  -\frac{dT^2}{(V+cT)^4} + h_{jk} dx^j dx^k.
\end{equation*}
For this metric, we can calculate the Christoffel symbols
\begin{equation*}
\Gamma^i_{00} = -\frac{2}{(V+cT)^5}h^{ij} \frac{\partial V}{\partial x^j} \,\,\,,\,\,\, \Gamma^i_{0j} = 0 = \Gamma^i_{jk}.
\end{equation*}
Hence, the geodesic equations for the spatial components may be written as
\begin{equation*}
\ddot{x}^i - \frac{2}{(V+cT)^5} h^{il}\frac{\partial V}{\partial x^l} \dot{T}^2 = F(s) \dot{x}^i \,\,\,\,\,;\,\,\,\,\, \dot{} = \frac{d}{ds}
\end{equation*}
where $F(s)$ is some function of the curve parameter $s$. If we assume that our geodesics are null, then we can rewrite this equation as
\begin{equation}
\ddot{x}^i - \frac{2}{V+cT} h^{il}\frac{\partial V}{\partial x^l} h_{jk} \dot{x}^j \dot{x}^k = F(s) \dot{x}^i.
\label{dudnullgeod}
\end{equation}
Now let us impose the condition that $s$ be the arc-length parameter for the metric $h$, i.e, $h_{jk} \dot{x}^j \dot{x}^k = 1$.
\begin{equation*}
\Rightarrow 0 = h_{jk} \dot{x}^j \ddot{x}^k = \frac{2}{V + cT} \frac{\partial V}{\partial x^k} \dot{x}^k + F(s)
\end{equation*}
\begin{equation*}
\Rightarrow F(s) = -\frac{2}{V + cT} \frac{\partial V}{\partial x^k} \dot{x}^k.
\end{equation*}
Hence, given that we use the arc-length parametrisation, the equation for null geodesics (\ref{dudnullgeod}) may be written as
\begin{equation}
\ddot{x}^i = \frac{2}{V+cT} \left( h^{il} \frac{\partial V}{\partial x^l} - \frac{\partial V}{\partial x^k} \dot{x}^k \dot{x}^i \right).
\label{propnullgeod}
\end{equation}
In the case when $h$ in (\ref{firsteq}) is non-flat, we can derive an analogous expression
\begin{equation}
\ddot{x}^i + \hat{\Gamma}^i_{jk} \dot{x}^j \dot{x}^k = \frac{2}{V+cT} \left( h^{il} \frac{\partial V}{\partial x^l} - \frac{\partial V}{\partial x^k} \dot{x}^k \dot{x}^i \right)
\label{twochelp}
\end{equation}
where $\hat{\Gamma}^i_{jk}$ are the connection components of the metric $h$. Two properties that we get from equation (\ref{propnullgeod}) are
\begin{eqnarray}
 h_{jk} \ddot{x}^j \ddot{x}^k &=& \frac{2}{V+cT} \frac{\partial V}{\partial x^k} \ddot{x}^k \label{iceprop} \\
\ddot{x}^k \frac{\partial V}{\partial x^k} &=& \frac{2}{V+cT} \left( \left( \frac{\partial V}{\partial x^k} \right) \left( \frac{\partial V}{\partial x_k} \right) - \left( \frac{\partial V}{\partial x^k} \dot{x}^k \right)^2 \right)
\label{niceprop}
\end{eqnarray}
Now, if we differentiate (\ref{propnullgeod}), we get the following
\begin{eqnarray*}
\dddot{x}^i &=& - \frac{2}{(V+cT)^2} \left( \frac{\partial V}{\partial x^k} \dot{x}^k + c \dot{T} \right) \left( h^{il}\frac{\partial V}{\partial x^l} - \frac{\partial V}{\partial x^k} \dot{x}^k \dot{x}^i \right) \\
&+& \frac{2}{V+cT} \left( h^{il} \frac{\partial^2 V}{\partial x^l \partial x^m} \dot{x}^m - \frac{\partial^2 V}{\partial x^k \partial x^l} \dot{x}^k \dot{x}^l \dot{x}^i - \frac{\partial V}{\partial x^k} \ddot{x}^k \dot{x}^i - \frac{\partial V}{\partial x^k} \dot{x}^k \ddot{x}^i\right).
\end{eqnarray*}
Then, using the null condition to eliminate $\dot{T}$ and equations (\ref{propnullgeod}), (\ref{iceprop}) and (\ref{niceprop}) to eliminate the $T$ parameter, we can rewrite this in the form:
\begin{eqnarray}
\dddot{x}^i &=& -  h_{jk} \ddot{x}^j \ddot{x}^k \dot{x}^i - \frac{3\ddot{x}^k \frac{\partial V}{\partial x^k}}{2\left(\frac{\partial V}{\partial x^k} \frac{\partial V}{\partial x_k} - \left( \frac{\partial V}{\partial x^k} \dot{x}^k \right)^2 \right)}\frac{\partial V}{\partial x^l} \dot{x}^l \ddot{x}^i - 2c \left( h^{il} \frac{\partial V}{\partial x^l} - \frac{\partial V}{\partial x^k} \dot{x}^k \dot{x}^i \right)  \nonumber \\
&+& \frac{2}{V+cT} \left( h^{il} \frac{\partial^2 V}{\partial x^l \partial x^m} \dot{x}^m - \frac{\partial^2 V}{\partial x^k \partial x^l} \dot{x}^k \dot{x}^l \dot{x}^i \right).
\label{thirdord}
\end{eqnarray}
Here, we can also eliminate the factor of $\frac{2}{V+cT}$ by using (\ref{niceprop}) so that we really have a system of third order ODEs completely dependent on the spatial coordinates of $g$ only. However, it is useful to keep it in the form above for the computation in the next section.

\section{One-Centre Case}
As was described in \cite{kastor}, a spacetime containing $N$ charged black holes with masses $m_\alpha$ ($\alpha = 1,\ldots,N$) and charges $q_\alpha = m_\alpha$ in a deSitter background can be represented by equation (\ref{firsteq}) with
\begin{equation*}
V = \displaystyle \sum_{\alpha=1}^N \frac{m_\alpha}{|\textbf{x}-\textbf{w}_\alpha|}
\end{equation*}
where $\textbf{w}_\alpha$ is a fixed vector for each $\alpha$. It is easily verified that $V$ is a harmonic function.
\\In this section, we look at the case $N=1$ where the black hole is situated at the origin. So, put $V = \frac{m}{|\textbf{x}|}$. With this definition, we can obtain the following identities:
\begin{equation*}
\frac{\partial V}{\partial x^j} = -\frac{m}{|\textbf{x}|^3} h_{jl} x^l \,\,\,\,\,,\,\,\,\,\, \frac{\partial^2 V}{\partial x^j \partial x^k} = \frac{3m}{|\textbf{x}|^5} h_{jl}x^j h_{km} x^m - \frac{m}{|\textbf{x}|^3} h_{jk}.
\end{equation*}
In particular, equation (\ref{propnullgeod}) becomes
\begin{equation}
\ddot{x}^i = \frac{2}{V+cT} \left( -\frac{m}{|\textbf{x}|^3} x^i + \frac{m}{|\textbf{x}|^3} (\textbf{x}.\dot{\textbf{x}}) \dot{x}^i \right),
\label{replacer}
\end{equation}
where $\textbf{x}.\dot{\textbf{x}} \equiv h_{jk} x^j \dot{x}^k$ and all subsequent dot products are taken with respect to the metric $h$ unless otherwise stated. This allows us to reduce the last term of (\ref{thirdord}), that is
\begin{eqnarray*}
& & \frac{2}{V+cT} \left( h^{il} \frac{\partial^2 V}{\partial x^l \partial x^m} \dot{x}^m - \frac{\partial^2 V}{\partial x^k \partial x^l} \dot{x}^k \dot{x}^l \dot{x}^i \right) \nonumber \\
&=& \frac{2}{V+cT} \left( \frac{3m}{|\textbf{x}|^5} (\textbf{x}.\dot{\textbf{x}}) x^i - \frac{3m}{|\textbf{x}|^5} (\textbf{x}.\dot{\textbf{x}})^2 \dot{x}^i \right) = -\frac{3}{|\textbf{x}|^2} (\textbf{x}.\dot{\textbf{x}}) \ddot{x}^i.
\end{eqnarray*}
With this simplification in mind, we can rewrite our system of third order ODEs (\ref{thirdord}) as
\begin{equation}
\dddot{x}^i = -|\ddot{\textbf{x}}|^2 \dot{x}^i + \frac{2mc}{|\textbf{x}|^3} \left( x^i - (\textbf{x}.\dot{\textbf{x}}) \dot{x}^i \right) - 3(\textbf{x}.\dot{\textbf{x}}) \left( \frac{1}{|\textbf{x}|^2} + \frac{\textbf{x}.\ddot{\textbf{x}}}{2(|\textbf{x}|^2 - (\textbf{x}.\dot{\textbf{x}})^2)} \right) \ddot{x}^i.
\label{mainequation}
\end{equation}
\\Any null geodesic of the Kastor-Traschen metric $g$ will project down to an integral curve of this system of third-order ODEs.

\subsection{Conformal Circles}
As $c \rightarrow 0$, it's obvious that the second term on the right-hand side of (\ref{mainequation}) vanishes. However, using the second equation in (\ref{niceprop}) with $c=0$ and $V = \frac{m}{|\textbf{x}|}$, we also find that
\begin{equation*}
\textbf{x}.\ddot{\textbf{x}} = 2 \left( \frac{(\textbf{x}.\dot{\textbf{x}})^2}{|\textbf{x}|^2} - 1 \right)
\end{equation*}
and the third term vanishes. To see the vanishing of the third term explicitly occuring with the vanishing of $c$, it seems we need to reintroduce the time coordinate $T$, in some way. For example, using (\ref{niceprop}) we can write (\ref{mainequation}) as
\begin{equation*}
\dddot{x}^i = -|\ddot{\textbf{x}}|^2 \dot{x}^i + \frac{2mc}{|\textbf{x}|^3} \left( x^i - (\textbf{x}.\dot{\textbf{x}}) \dot{x}^i \right) + 3(\textbf{x}.\dot{\textbf{x}}) \left( \frac{cT|\textbf{x}|(\textbf{x}.\ddot{\textbf{x}})}{2m(|\textbf{x}|^2-(\textbf{x}.\dot{\textbf{x}})^2)} \right) \ddot{x}^i.
\end{equation*}
Hence, as $c \rightarrow 0$, null geodesics satisfy
\begin{equation}
\dddot{x}^i = -|\ddot{\textbf{x}}|^2 \dot{x}^i.
\label{confcirc}
\end{equation}
We shall see that this system (\ref{confcirc}) occupies a central role in the theory of conformally flat manifolds.
\\In general, given a conformal structure $[\tilde{h}]$ on an $n$-dimensional manifold, there is a distinguished family of curves, known as the \emph{conformal circles}. These curves arise as the integral curves of a system of third order ODEs, see \cite{bailey}. To write this system down, let us choose a metric $\tilde{h}$ in the conformal class with torsion-free connection $\tilde{\Gamma}^i_{jk} = \tilde{\Gamma}^i_{(jk)}$, Ricci tensor $R_{jk}$ and scalar curvature $R$. Then, the Schouten tensor $P_{jk}$ is defined as
\begin{equation*}
P_{jk} = - \frac{1}{n-2} \left( R_{jk} - \frac{R}{2(n-1)} \tilde{h}_{jk} \right).
\end{equation*}
Furthermore, define the vector components $U^i = \dot{x}^i$ and $A^i = \ddot{x}^i + \tilde{\Gamma}^i_{jk} \dot{x}^j \dot{x}^k$ so that $\textbf{U}.\textbf{A} = \tilde{h}_{jk} U^j A^k$, etc, and scalar products coincide with our previous definitions when $\tilde{h}$ is flat. Then, a curve is a conformal circle of $[\tilde{h}]$ if it satisfies
\begin{equation}
\frac{dA^i}{ds} + \tilde{\Gamma}^{i}_{jk}A^j U^k = \frac{3 \textbf{U}.\textbf{A}}{|\textbf{U}|^2} A^i - \frac{3|\textbf{A}|^2}{2|\textbf{U}|^2} U^i + |\textbf{U}|^2 U^j P_{j}^i - 2P_{jk}U^j U^k U^i
\label{baileycc}
\end{equation}
where there is no restriction on the parameter $s$. Equation (\ref{baileycc}) is invariant with respect to conformal transformations, $\tilde{h} \rightarrow \Omega^2 \tilde{h}$, and so, conformal circles are defined invariantly by any metric in the conformal class. These curves have arisen in a physical context in \cite{friedschmid} where the authors have used them to discuss the asymptotics of Einstein's equations. Furthermore, properties of ``conformal geodesics'' (lifts of conformal circles to the bundle of second order frames over the manifold endowed with the conformal Cartan connection) in vacuum and warped-product spacetimes have been studied in \cite{friedrich}.
\\It is shown in $\cite{bailey}$ that the conformal circles of a given conformal manifold can be equally defined as the set of integral curves of the system of ODEs
\begin{equation}
\frac{dA^i}{ds} + \tilde{\Gamma}^{i}_{jk}A^j U^k = - |\textbf{A}|^2 U^i + U^j P_j^i - P_{jk} U^j U^k U^i
\label{moreuseful}
\end{equation}
where, here, $s$ is required to be the arc-length parameter of the metric $\tilde{h}$. This formulation was originally given by Yano in \cite{yanocc} and is more useful for our purposes. If we now let $\tilde{h}$ be the flat Riemannian metric, equation (\ref{moreuseful}) reduces to (\ref{confcirc}).
\\Hence, as $c \rightarrow 0$ in the one-centre Kastor-Traschen metric, null geodesics project down to conformal circles of the flat metric in three dimensions.  It is easily verified, that the set of integral curves of (\ref{confcirc}) in three dimensions is precisely the set of all circles in $\mathbb{R}^3$.
\\We should note here that we get the same result when we let $m \rightarrow 0$. Again, the second term on the right-hand side of (\ref{mainequation}) obviously vanishes and (\ref{replacer}) reduces to
\begin{equation*}
\ddot{x}^i = 0,
\end{equation*}
so the third term also vanishes.

\section{Characterisation of Null geodesics}
In the preceding sections, we have determined a system of third-order ODEs which the projected null geodesics of $g$ along the conformal retraction $\Theta$ must satisfy. However, it is not clear that, given an integral curve of (\ref{thirdord}), it will necessarily be the projection of some light ray of $g$. In fact, we can show this not to be the case and it transpires that the third order ODE system (\ref{thirdord}) is not uniquely defined. In this section, we discuss this point and construct a new third order system in three dimensions for which the integral curves constitute a retraction projection of a special set of null curves of the Kastor-Traschen metric, which have a physical interpretation. The projected null geodesics form a subset of these curves which we can characterize.
\\\\For example, let us consider the case $c \rightarrow 0$ for the one-centre metric. Here, we found that the integral curves of (\ref{confcirc}) will be the set of circles in $\mathbb{R}^3$.
\\However, as $c \rightarrow 0$, the metric $g$ becomes static with static Killing vector $\Theta = \frac{\partial}{\partial T}$ and it is well known that the null geodesics of this metric project down to the unparametrised geodesics of the associated optical metric
\begin{equation*}
h_{\text{opt}} = \left(\frac{m}{|\textbf{x}|} \right)^4 h_{jk} dx^j dx^k.
 \end{equation*}
 One can check that the unparametrised geodesics of this metric will be precisely the set of circles in $\mathbb{R}^3$ which pass through the origin. Hence, only a subset of the integral curves of (\ref{confcirc}) will coincide with the projected null geodesics of $g$. Note that this is also the case for $m \rightarrow 0$ where projected null geodesics are described by $\ddot{x}^i = 0$ - straight lines. We can check the general numerology here to see what happens.
\\\\Firstly, we note that the set of unparametrised geodesics of an arbitrary metric $\hat{g}$ on some open set $U \subset \mathbb{R}^n$ will lift to a foliation, by the geodesic spray, of the projectivised tangent bundle $\mathbb{P}(TU)$ which we can think of as a 1-dimensional fibration over some $(2n-2)$-dimensional space, $Z$, with each point in $Z$ coinciding with a unique geodesic in $U$. Hence, the set of unparametrised geodesics of an $n$-dimensional manifold constitute a $(2n-2)$-parameter family of curves.
\\Taking the specific example of the Kastor-Traschen metric $g$, we have $n=4$ and so the number of parameters describing unparametrised geodesics is 6. Invoking the null condition, we see that the retraction projection of unparametrised null geodesics will, in general, constitute a 5-parameter family of curves in $\mathbb{R}^3$ - in the special case where the conformal retraction is a static Killing vector, this is a 4-parameter family.
\\\\On the other hand, let us consider a set of curves on some open set $U \subset \mathbb{R}^n$ described by a system of third order ODEs. If we write this in an unparametrised way - as a set of $(n-1)$ third-order ODEs in terms of one of the coordinates - then we see that the integral curves of this system will lift to a foliation of the jet bundle $J^2(U,\mathbb{R})$ which is $(3n-2)$-dimensional.
\\Hence the unparametrised integral curves of a system of third-order ODEs in $n$ dimensions constitutes a $(3n-3)$-parameter family of paths. When $n=3$, for example, we will have a 6-parameter family of such curves which is consistent with our results above.
\\\\Hence, the set of projected null geodesics of the Kastor-Traschen metric, $g$, will constitute a 5-dimensional subset of the 6-dimensional family of unparametrised integral curves of (\ref{thirdord}) (except in the static case).
\\So, a natural question arises: Given that we construct a third order system (such as (\ref{thirdord})) for which the projected null geodesics form a proper subset of the set of integral curves then to what do the other integral curves correspond?
\\The system (\ref{thirdord}) does not help us to answer this question but we can derive a different third order system which will. For convenience and clarity on this point, let us write our system of equations describing null geodesics of the Kastor-Traschen metric (\ref{propnullgeod}) in three-dimensional vector notation i.e,
\begin{equation*}
\ddot{\textbf{x}} = \frac{2}{V(\textbf{x}) + c T} \left( \nabla V - (\nabla V.\dot{\textbf{x}})\dot{\textbf{x}} \right) = \frac{2}{V+cT} \left( \dot{\textbf{x}} \times (\nabla V \times \dot{\textbf{x}}) \right).
\end{equation*}
Now let us consider a modification of this equation by adding an orthogonal term on the right-hand side, that is
\begin{equation}
\ddot{\textbf{x}} = \frac{2}{V+cT} \left( \dot{\textbf{x}} \times (\nabla V \times \dot{\textbf{x}}) \right) + \lambda (\dot{\textbf{x}} \times \nabla V)
\label{magfield}
\end{equation}
where $\lambda$ is a constant. Clearly, null geodesics satisfy this equation for $\lambda = 0$. More interestingly, there is a six-parameter family of curves which satisfy this equation for some value of $\lambda$. Hence, we might expect these curves to be the integral curves of some third order system in three dimensions which is independent of $\lambda$.
\\First from ({\ref{magfield}), we can derive the following equations by taking specific scalar products:
\begin{eqnarray}
|\ddot{\textbf{x}}|^2 &=& \frac{2}{V+cT} \nabla V.\ddot{\textbf{x}} + \lambda (\dot{\textbf{x}} \times \nabla V). \ddot{\textbf{x}} \nonumber\\
\ddot{\textbf{x}}.\nabla V &=& \frac{2}{V+cT} |\nabla V \times \dot{\textbf{x}} |^2 \nonumber \\
\lambda &=& \frac{\ddot{\textbf{x}}.(\dot{\textbf{x}} \times \nabla V)}{|\dot{\textbf{x}} \times \nabla V |^2}.
\label{diffsys}
\end{eqnarray}
Differentiating (\ref{magfield}) and simplifying using the first two equations of (\ref{diffsys}), we derive the system of third order ODEs
\begin{eqnarray*}
\dddot{\textbf{x}} &=& -|\ddot{\textbf{x}}|^2 \dot{\textbf{x}} -\frac{3}{2} \frac{\ddot{\textbf{x}}.\nabla V}{|\dot{\textbf{x}} \times \nabla V|^2}(\dot{\textbf{x}}.\nabla V) \ddot{\textbf{x}} - 2c (\dot{\textbf{x}} \times (\nabla V \times \dot{\textbf{x}})) + \frac{\ddot{\textbf{x}}.\nabla V}{|\dot{\textbf{x}} \times \nabla V|^2} \left( \dot{\textbf{x}} \times \left( \frac{d \nabla V}{ds} \times \dot{\textbf{x}} \right) \right) \\
&+& \lambda \left[ \left( \frac{\ddot{\textbf{x}}.\nabla V}{|\dot{\textbf{x}} \times \nabla V|^2} \right) (\dot{\textbf{x}}.\nabla V) \ddot{\textbf{x}} \times \nabla V + ((\dot{\textbf{x}} \times \nabla V).\ddot{\textbf{x}})\dot{\textbf{x}} + \ddot{\textbf{x}} \times \nabla V + \dot{\textbf{x}} \times \frac{d \nabla V}{ds} \right].
\end{eqnarray*}
For $\lambda = 0$, this system of equations reduces to (\ref{thirdord}) with the $\frac{2}{V+cT}$ term replaced using (\ref{niceprop}) as expected. We can eliminate $\lambda$ from this equation using the third equation of (\ref{diffsys}). Notice then that the vanishing of $\lambda$ coincides with $\ddot{\textbf{x}}.(\dot{\textbf{x}} \times \nabla V) = 0$ i.e, the vectors $\ddot{\textbf{x}}, \dot{\textbf{x}}$ and $\nabla V$ lie in the same plane. In particular, for the one centre case, the projections of null geodesics lie in a plane through the origin (centre).
Overall, we have the following result
\\\\\textbf{Proposition 5.1} \emph{If $c \neq 0$, the retraction projection of the set of null curves satisfying (\ref{magfield}) for some value of $\lambda$ coincides with the set of integral curves of}
\begin{eqnarray}
\dddot{\textbf{x}} &=& -|\ddot{\textbf{x}}|^2 \dot{\textbf{x}} -\frac{3}{2} \frac{\ddot{\textbf{x}}.\nabla V}{|\dot{\textbf{x}} \times \nabla V|^2}(\dot{\textbf{x}}.\nabla V) \ddot{\textbf{x}} - 2c (\dot{\textbf{x}} \times (\nabla V \times \dot{\textbf{x}})) + \frac{\ddot{\textbf{x}}.\nabla V}{|\dot{\textbf{x}} \times \nabla V|^2} \left( \dot{\textbf{x}} \times \left( \frac{d \nabla V}{ds} \times \dot{\textbf{x}} \right) \right) \nonumber \\
&+& \frac{\ddot{\textbf{x}}.(\dot{\textbf{x}} \times \nabla V)}{|\dot{\textbf{x}} \times \nabla V |^2} \left[ \left( \frac{\ddot{\textbf{x}}.\nabla V}{|\dot{\textbf{x}} \times \nabla V|^2} \right) (\dot{\textbf{x}}.\nabla V) \ddot{\textbf{x}} \times \nabla V + ((\dot{\textbf{x}} \times \nabla V).\ddot{\textbf{x}})\dot{\textbf{x}} + \frac{d}{ds} (\dot{\textbf{x}} \times \nabla V) \right].
\label{bigsys}
\end{eqnarray}
\emph{Furthermore, the projected null geodesics are precisely the integral curves of this system for which} $\ddot{\textbf{x}}.(\dot{\textbf{x}} \times \nabla V) = 0$.

\begin{proof}
As we have shown, any integral curve of (\ref{magfield}) satisfies (\ref{bigsys}). To verify the reverse inclusion, we just need to consider the initial data unique to one integral curve $\gamma$ of (\ref{bigsys}) which will be given by seven parameters - three for initial position, two for initial unit velocity and two for initial acceleration (perpendicular to the velocity vector). By varying the values of $T$ and $\lambda$, it is clear that there is an integral curve of (\ref{magfield}) with the same initial data and its projection necessarily coincides with $\gamma$.
\end{proof}
The proposition doesn't work for $c = 0$ as we cannot use $T$ as a parameter for the initial acceleration data in the above proof.

\subsection*{Magnetic Flow}
The addition of this extra $\lambda$ term may seem a little \emph{ad hoc} here but is actually a sensible choice when we see the proof of this proposition i.e, we need to add a term orthogonal to $\dot{\textbf{x}}$ but not in the direction of $\dot{\textbf{x}} \times (\nabla V \times \dot{\textbf{x}})$. Furthermore, this system of equations (\ref{magfield}) can be interpreted as describing a magnetic flow in the background of the Kastor-Traschen metric with magnetic field $\textbf{B} \propto \nabla V$. This is precisely the magnetic field we saw in (\ref{newf}) when discussing the freedom in the 2-form $F$ and so these additional integral curves occupy a significant role in the geometry of the Kastor-Traschen metric.

\subsection{A Solution in the One-Centre Case}
Let $\varphi = 4mc$ and define a curve in the plane $x^3 = 0$ by
\begin{equation}
x^i (s) = \left( \varphi s \cos \left( \frac{\sqrt{1-\varphi^2}}{\varphi} \log (\varphi s) \right), \varphi s \sin \left( \frac{\sqrt{1-\varphi^2}}{\varphi} \log (\varphi s) \right), 0 \right).
\label{samplenull}
\end{equation}
Then this curve satisfies $h_{jk} \dot{x}^j \dot{x}^k = 1$ and is an integral curve of the system of ODEs (\ref{bigsys}) for $V = \frac{m}{|\textbf{x}|}$. Furthermore, since it lies on a plane through the origin, we know, by Proposition 5.1, that it must be the retraction projection of a null geodesic of $g$.
\\We can plot this curve in the plane and realise that it is just a reparametrised logarithmic spiral.
\begin{figure}[h]
\centering \includegraphics[width=4cm,height=3cm]{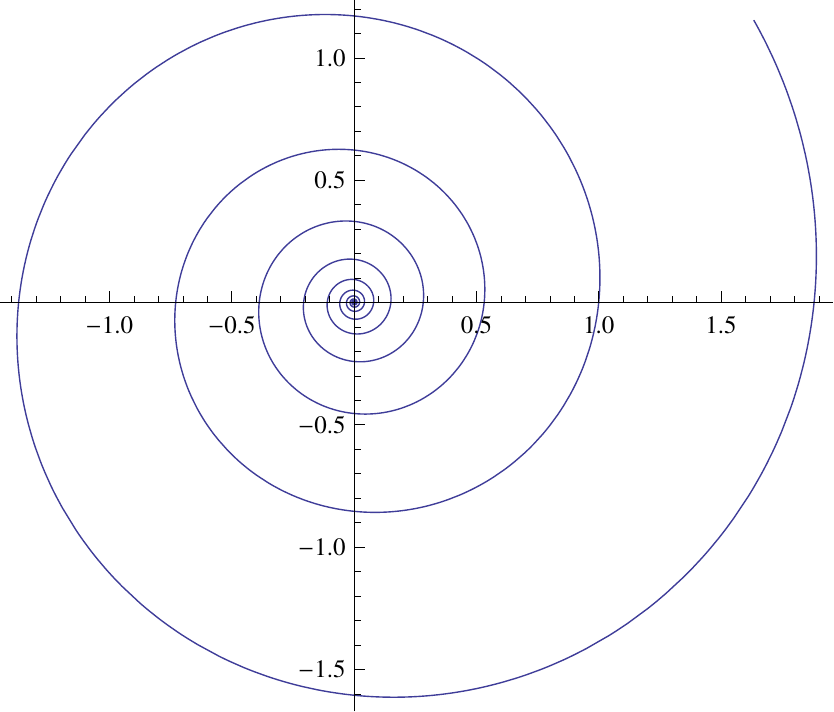}
\caption{Logarithmic Spiral with $\varphi = 0.1$}
\end{figure}
This example is motivated by work in section 6 where we will show how to derive analytic expressions for the projected null geodesics of the one-centre K-T metric and thus, the integral curves of (\ref{mainequation}) which lie on a plane through the origin. We should note here, that the limit $c \rightarrow 0$, for this example, is ill-defined. By L'Hopital's Rule, both expressions in (\ref{samplenull}) tend to zero, in this limit, for any value of $s$.

\section{Geodesics obtained from Extremal RNdS}
In this section, we will make use of a diffeomorphism between the one-centre Kastor-Traschen and the extremal Reissner-Nordstrom deSitter metrics. The advantage of this is that analytic solutions for the null geodesic equations of the RNdS metrics are well known \cite{hackmann}. We will show how to derive these solutions which will enable us to obtain analytic solutions in the Kastor-Traschen coordinates and plot the retraction projection of the null geodesics in some cases.

\subsection{RNdS Transformation}
We begin with the special Kastor-Traschen metric with potential $V = \frac{m}{|\textbf{x}|}$ and now use spherical polar coordinates ($|\textbf{x}| = R$) to represent the flat metric $h$ i.e,
\begin{equation}
g = - \frac{1}{\left(\frac{m}{R} + c T \right)^2}dT^2 + \left( \frac{m}{R} + cT \right)^2 (dR^2 + R^2(d\theta^2 + \sin^2 \theta d \phi^2)).
\label{KTform}
\end{equation}
Assuming that $c \neq 0$ we can make the coordinate transformation
\begin{equation}
R = e^{-cS} \,\,\,\,\,,\,\,\,\,\, T = \frac{r-m}{c} e^{cS}.
\label{transforig}
\end{equation}
If we choose $t$ such that $dt = dS + \frac{r-m}{c \Delta_r} dr$, where $\Delta_r = (r-m)^2 - c^2 r^4$, then the metric becomes
\begin{equation}
g = - \frac{\Delta_r}{r^2} dt^2 + \frac{r^2}{\Delta_r} dr^2 + r^2(d\theta^2 + \sin^2 \theta d\phi^2).
\label{RNdSmet}
\end{equation}
The analogous coordinate transformation in the Riemannian case was given in \cite{DGST}. Here, the resulting metric is the Reissner-Nordstrom deSitter spacetime in the extremal case with charge $Q = m$ and $c = \pm \sqrt{\frac{\Lambda}{3}}$ where $\Lambda$ is the cosmological constant. From now on, we will refer to $t$ as \emph{static time} in order to distinguish it from the time $T$. We also take $c > 0$ as the definition of (\ref{RNdSmet}) is invariant with respect to $c \rightarrow -c$ and we can compensate for it in the Kastor-Traschen case by sending $T \rightarrow -T$.
\\Null and timelike geodesics of black hole spacetimes with cosmological constant were studied extensively in \cite{hackmann}. In particular, the author discussed the different types of orbits possible for Reissner-Nordstrom metrics and showed how to derive analytic formulae for the geodesics. The RNdS metric admits the timelike static Killing vector $\frac{\partial}{\partial t}$ and we can plot the projection of the null geodesics to the space of orbits of $\frac{\partial}{\partial t}$, analogous to what was done in \cite{hackmann}, which we call the \emph{static projection}. We then use this information to plot the null geodesics in the Kastor-Traschen metric projected along the conformal retraction $\Theta$, which, by Proposition 5.1, will be solutions of the system (\ref{bigsys}) and lie on a plane through the origin.
\\Null geodesics of the RNdS metric (\ref{RNdSmet}) satisfy
\begin{equation}
-\frac{\Delta_r}{r^2} \dot{t}^2 + \frac{r^2}{\Delta_r} \dot{r}^2 + r^2 (\dot{\theta}^2 + \sin^2\theta \, \dot{\phi}^2) = 0.
\label{null}
\end{equation}
The Euler-Lagrange equation for $\theta$ gives
\begin{equation*}
\frac{d}{d s}(2r^2 \dot{\theta}) = 2 r^2 \sin \theta \cos \theta \dot{\phi}^2,
\end{equation*}
where $s$ parameterizes the curves. By a choice of axes, we can set the initial conditions to be $\dot{\theta} = 0$, $\theta = \frac{\pi}{2}$, which results in motion in the equatorial plane - this coincides with the results of Proposition 5.1. Similarly, from the Euler-Lagrange equation for $\phi$ and $t$ we find that
\begin{eqnarray*}
r^2 \dot{\phi} &=&  \Phi \\
- \frac{2 \Delta_r}{r^2} \dot{t} &=& -2E
\end{eqnarray*}
where $\Phi$ and $E$ are constants. Let us focus our attention on non-radial geodesics and assume that $\Phi > 0$ so that the geodesics are traced out in the direction of increasing $\phi$. The null equation (\ref{null}) implies that
\begin{equation*}
- \frac{r^2}{\Delta_r} E^2 + \frac{r^2}{\Delta_r} \dot{r}^2 + \frac{\Phi^2}{r^2} = 0
\end{equation*}
or
\begin{equation}
\dot{r}^2 = E^2 - \frac{\Delta_r}{r^4} \Phi^2 \equiv E^2 - V_{eff}.
\label{veffeq}
\end{equation}
Here $V_{eff}$ is the effective potential of the system which we can plot as a function of $r$.
\begin{figure}[h]
\centering \includegraphics[width=6cm,height=3cm]{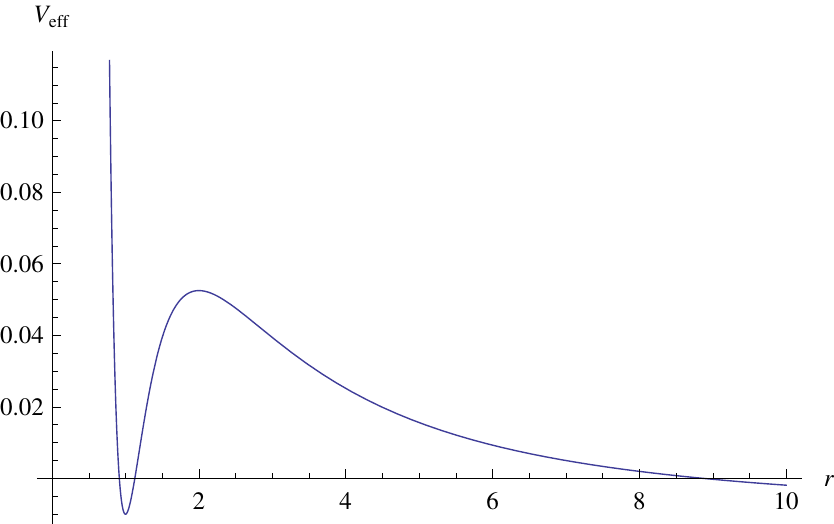}
\caption{Plot of $V_{eff}$ as a function of $r$}
\end{figure}
\\We note here that physically acceptable regions for null geodesics are those for which $E^2 \geq V_{eff}$. From the diagram, it is clear that we have two kinds of orbits - \emph{bound orbits} (where $r$ oscillates between two boundary values) and unbound \emph{flyby orbits} (where $r$ starts at $\infty$, then approaches a periapsis and goes back to $\infty$). This is consistent with the results of $\cite{hackmann}$ and we use the same terminology.

\subsection{Circular Orbits}
As was observed in \cite{stuchlik}, the extremum values of $V_{eff}$ occur at $r = m$ and $r=2m$, independent of the cosmological constant (in this case, independent of $c$), resulting in circular orbits. The minimum of this function will not always be non-negative and, since $E^2 \geq 0$, it will not be attained for some values of $c$. Indeed, unless $c=0$, we will not get a circular orbit at $r=m$. However, the transformation for $R$ and $T$ does not behave well in this limit and therefore does not allow us to see what this circular orbit corresponds to in the Kastor-Traschen coordinates.
\\On the other hand, we will observe a circular orbit at $r=2m$ as long as $4mc \leq 1$ (beyond this, the local maximum drops below the axis in Figure 2). For this solution, there will be some factors of $c$ in the static time variable $t$ and hence, when we make the coordinate transformation to $R$ and $T$, the orbit in the retraction projection will be dependent on $c$. In fact, when we do this, the resulting orbit is just essentially a constant multiple of that given section 5 - a logarithmic spiral in some hyperplane of the space of orbits of the conformal retraction $\Theta$, which passes through the origin.

\subsection{Horizon Structure and Nature of Orbits}
Clearly, for $E^2 > V_{max}$, all orbits are unbound in the RNdS coordinates whereas for $E^2 < V_{max}$, we can get both bound and unbound orbits. If $m = 0$ then this graph becomes $V_{eff} = \frac{1}{r^2} + const$, and there are no bound orbits.
\\Using the chain rule and the E-L equation for $\phi$, we can rewrite equation (\ref{veffeq}) as
\begin{equation}
\left( \frac{dr}{d\phi} \right)^2 = (\kappa^2 + c^2)r^4 - (r-m)^2
\label{maineq}
\end{equation}
where $\kappa = \frac{E}{\Phi}$. Also,
\begin{equation*}
\frac{dt}{d\phi} = \kappa \frac{r^4}{\Delta_r}.
\end{equation*}
Note here that roots of $\Delta_r$ (which we inevitably cross for some orbits) will cause infinities in $\frac{dt}{d\phi}$ and in $t$ itself - this will lead to a null geodesic tracing out a finite path in an infinite amount of static time (see analysis). Radii at which $V_{eff} = 0$, or equally $\Delta_r = 0$, correspond to horizons. In particular, if:
\begin{itemize}
\item $4mc < 1$; there are three horizons (two black hole horizons and one cosmological horizon) with $r_{bh-} < r_{bh+} < r_{ch}$. The geometry is static for $r < r_{bh-}$, $r_{bh+} < r < r_{ch}$ and corresponds to a black hole in a de-Sitter universe.
\item $4mc > 1$; there is only one cosmological horizon. The geometry is static for $r < r_{ch}$ and corresponds to a naked singularity in a de-Sitter universe.
\end{itemize}
The metric (\ref{RNdSmet}) has a singularity at $r=0$ which is covered by the Cauchy horizon $r = r_{bh-}$ in the case $4mc < 1$. The surface gravity at this inner horizon is larger than that at the cosmological horizon, in particular,
\begin{equation*}
\kappa_{bh-}^2 - \kappa_{ch}^2 = 8mc^3 > 0
\end{equation*}
and so, by a result in \cite{mellor}, the Cauchy horizon is unstable. Hence, some of the trajectories in the Reissner-Nordstrom-deSitter metric will be unphysical (in particular when $r < r_{bh-}$). Furthermore the case $4mc > 1$ presents a possible violation of Penrose's cosmic censorship conjecture \cite{penrose} and may therefore also be unphysical.
\\It was noted in \cite{romans} that, for the Reissner-Nordstrom deSitter metric in the extreme charge equal mass case, the Hawking temperature of the outer black hole horizon is the same, in magnitude, as that of the cosmological horizon endowing a notion of thermodynamic stability among all RNdS solutions.
\\The plots of projected null geodesics can be obtained in the Kastor-Traschen framework by solving equation (\ref{maineq}) and making a coordinate transformation $(t,r) \rightarrow (T,R)$. As an example, we will perform this calculation for bound orbits in the three-horizon case (when $4mc <1$). The other curves can be similarly obtained by the reader but do not add much to the discussion. In the following example, we highlight when a trajectory is physical or when it is purely of mathematical interest.
\subsection*{Bound Orbits in Kastor-Traschen with Three Horizons}
This case corresponds to the inequalities
\begin{equation*}
4m\lambda < 1 \,\,\,\,\,,\,\,\,\,\, \frac{-1+\sqrt{1+4m\lambda}}{2\lambda} \leq r \leq \frac{1-\sqrt{1-4m\lambda}}{2\lambda}.
\end{equation*}
For bound orbits, any solution of (\ref{maineq}) will oscillate between the two extremal values for $r$ given above. Beginning with this equation, we can rearrange to get
\begin{equation}
\pm \int \frac{dr}{\sqrt{\lambda^2 \left( r^4- \left(\frac{r-m}{\lambda} \right) \right)}} = \phi + \gamma,
\label{rphirel}
\end{equation}
where $\gamma$ is a constant of integration. With the given bounds on $r$, we can integrate the left hand side of this equation and obtain
\begin{equation*}
\mp \frac{2 F \left( \arcsin \left(\sqrt{\frac{2 \sqrt{1-4m\lambda}(-1-2r\lambda+\sqrt{1+4m\lambda})}{(1-2r\lambda+\sqrt{1-4m\lambda})(-2+\sqrt{1+4m\lambda}+ \sqrt{1-4m\lambda})}}\right) , \frac{1}{2} - \frac{1}{2\sqrt{1-16m^2 \lambda^2}} \right)}{(1-16 m^2\lambda^2)^{1/4}} = \phi + \gamma
\end{equation*}
where $F$ is the elliptic integral of the first kind. We can solve this equation for $r = r(\phi)$ to get
\begin{equation*}
r(\phi) = \frac{\sin^2(J_{\mp}(\phi))(1+\sqrt{1-4m\lambda})(\sqrt{1+4m\lambda}+\sqrt{1-4m\lambda}-2)+2\sqrt{1-4m\lambda}(1-\sqrt{1+4m\lambda})}{2\lambda(\sin^2(J_{\mp}(\phi))(\sqrt{1+4m\lambda}+\sqrt{1-4m\lambda}-2)-2\sqrt{1-4m\lambda})}
\end{equation*}
where
\begin{equation*}
J_{\mp} (\phi) = \text{Jac} \left( \mp \frac{(1-16m^2\lambda^2)^{1/4}}{2}(\phi + \gamma), \frac{1}{2} - \frac{1}{2\sqrt{1-16m^2\lambda^2}} \right)
\end{equation*}
and Jac is the Jacobi Amplitude for the elliptic integral (i.e, $F(a,b) = c \Rightarrow a = \text{Jac}(c,b)$). To obtain the oscillatory solution we use $J_{(-1)^n}(\phi)$ whenever
\begin{equation*}
-\frac{2 F\left( \frac{n-1}{2}\pi , \frac{1}{2} - \frac{1}{2\sqrt{1-16m^2\lambda^2}}\right)}{(1-16m^2\lambda^2)^{1/4}} > \phi + \gamma \geq - \frac{2 F\left( \frac{n}{2} \pi , \frac{1}{2} - \frac{1}{2\sqrt{1-16m^2\lambda^2}}\right)}{(1-16m^2\lambda^2)^{1/4}}.
\end{equation*}
We can now plot the static projection of the null geodesics of the RNdS metric. For this purpose, we take the values $m=1$, $\kappa = \frac{1}{6}$, $c = \frac{1}{8}$ and $\gamma = 0$.
\begin{figure}[h]
\begin{minipage}[b]{0.5\linewidth}
\centering \includegraphics[width=3.6cm,height=3.6cm]{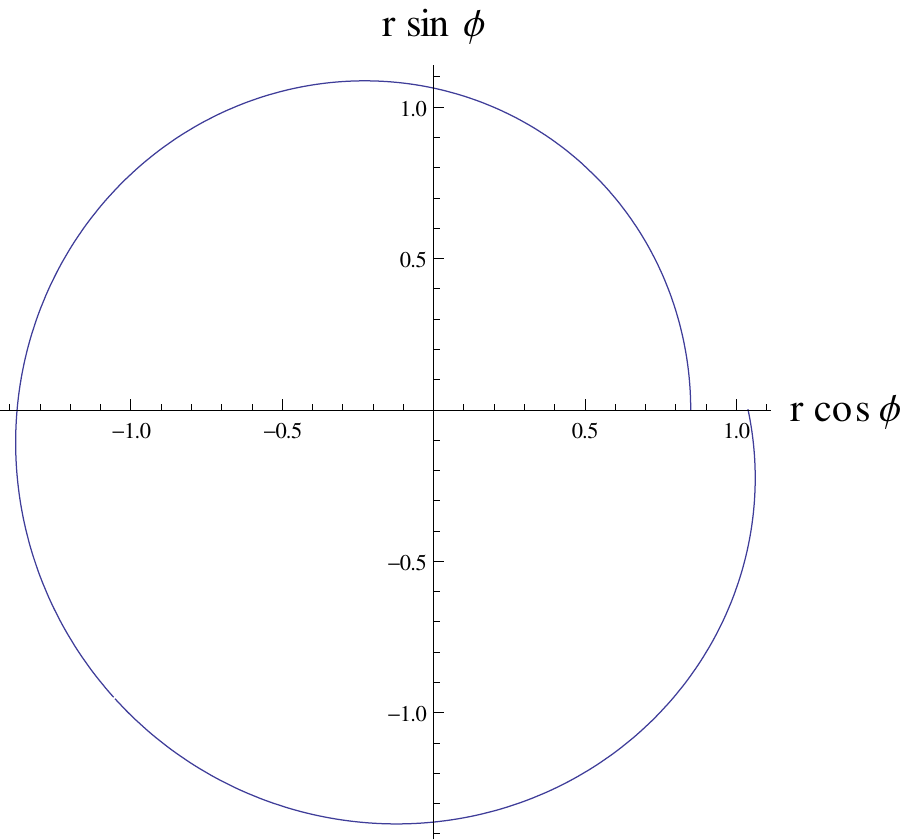}
\caption{Static projection of null geodesic, $0\leq \phi \leq 2\pi$}
\end{minipage}
\hspace{0.5cm}
\begin{minipage}[b]{0.5\linewidth}
\centering \includegraphics[width=3.6cm, height=3.6cm]{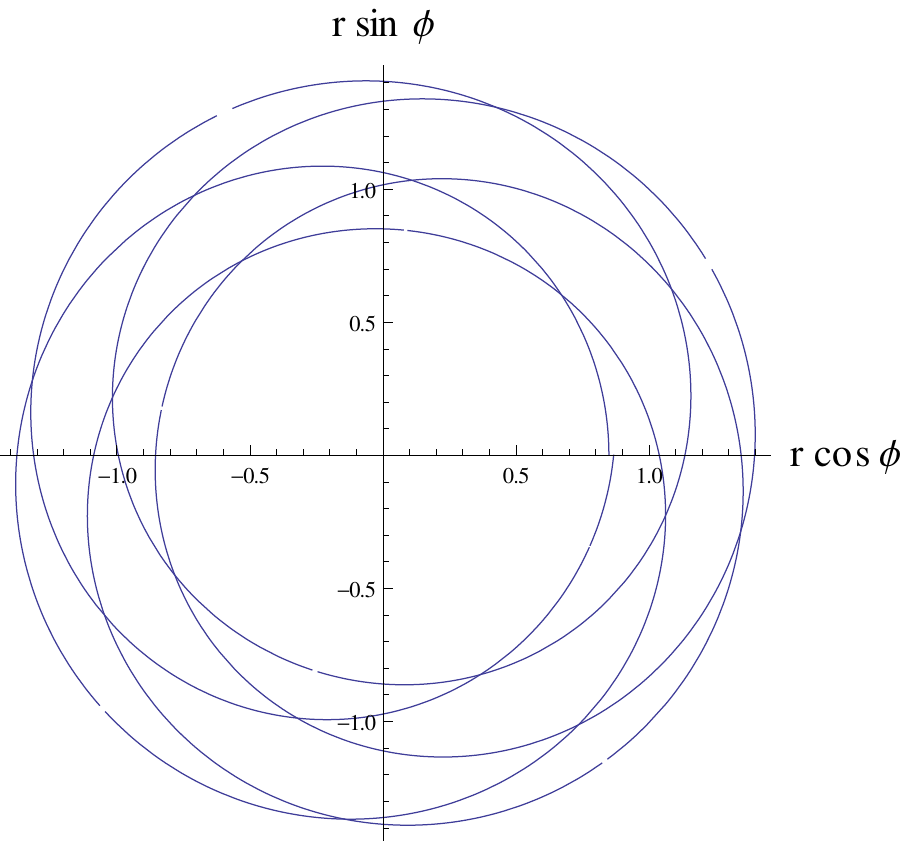}
\caption{$0\leq \phi \leq 10\pi$}
\end{minipage}
\end{figure}
\\As we mentioned before, it is important that we be careful here with respect to the range of the static time coordinate. In the above plots, the function $r$ will, at several stages, cross a value for which $\Delta_r = 0$ and satisfy $r < r_{bh-}$, where the trajectory is unphysical. This is reflected by the fact that each horizon crossing leads to an infinity in the static time $t$. For example, if we take a segment of this orbit which passes from the $r_{bh-}$ to $r_{bh+}$, we obtain the following plots for the geodesic itself and $t$ as a function of $\phi$ on this range.
\begin{figure}[h!]
\begin{minipage}[b]{0.5\linewidth}
\centering \includegraphics[width=4.5cm,height=4.5cm]{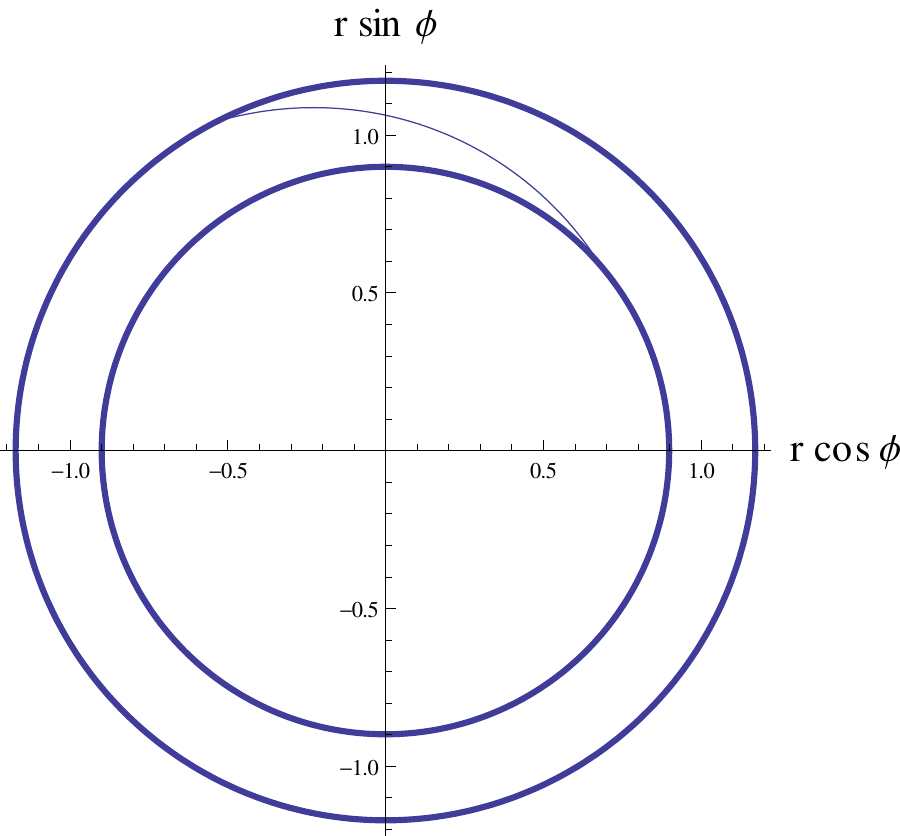}
\caption{Null geodesic traversing between black hole horizons $r_{bh-}$ and $r_{bh+}$.}
\end{minipage}
\hspace{0.5cm}
\begin{minipage}[b]{0.5\linewidth}
\centering \includegraphics[width=4.5cm, height=4.5cm]{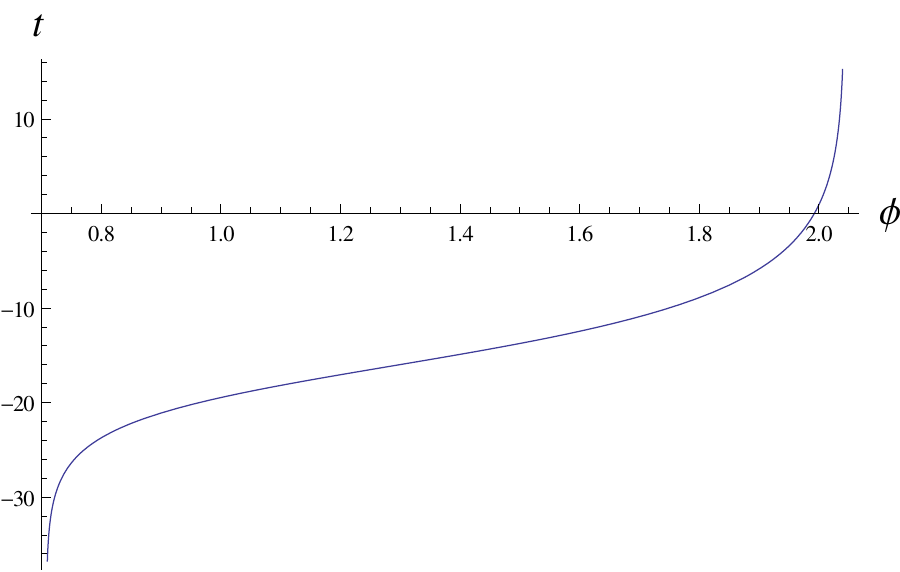}
\caption{Static time $t$ as a function of $\phi$ on this interval.}
\end{minipage}
\end{figure}
\newpage
Using the transformation (\ref{transforig}) together with the subsequent one for $S$, we can plot the retraction projection of this null geodesic in the Kastor-Traschen coordinates. Similarly, we can determine the time $T$ as a function of $\phi$ and we get the following plots:
\begin{figure}[h!]
\begin{minipage}[b]{0.5\linewidth}
\centering \includegraphics[width=4.5cm,height=4.5cm]{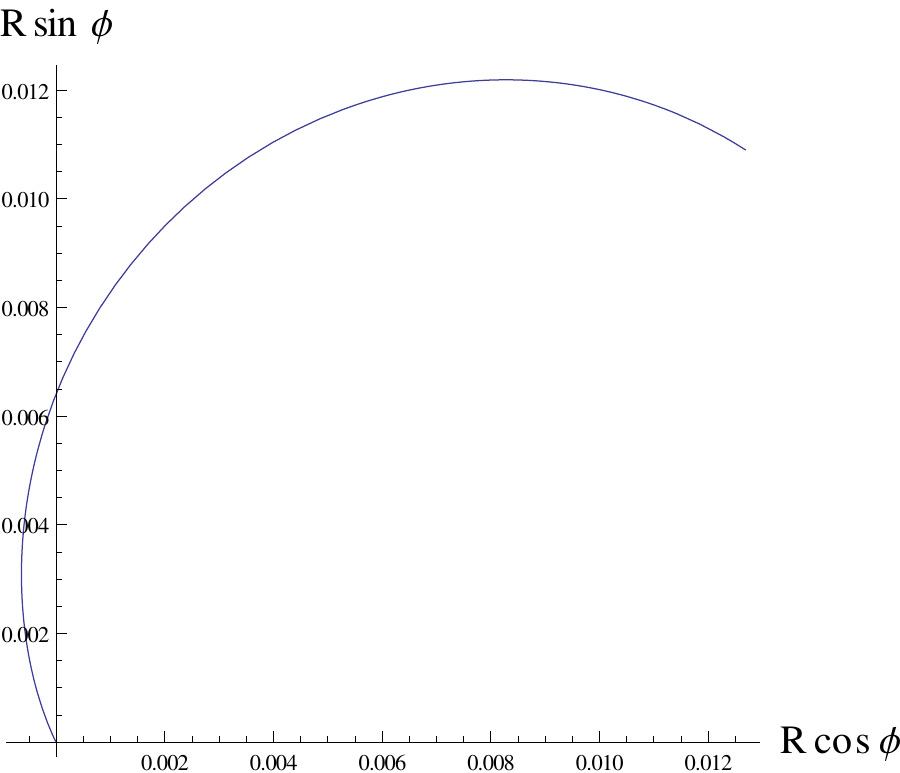}
\caption{Retraction projection of null geodesic of the K-T metric}
\end{minipage}
\hspace{0.5cm}
\begin{minipage}[b]{0.5\linewidth}
\centering \includegraphics[width=4.5cm, height=4.5cm]{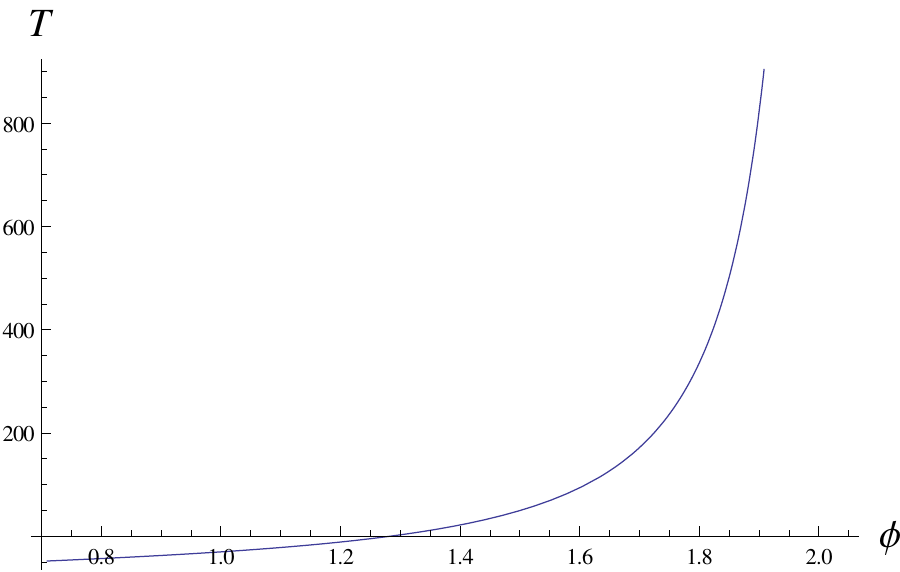}
\caption{Time $T$ as a function of $\phi$ on this interval.}
\end{minipage}
\end{figure}
\\As $R \rightarrow 0$, the time $T \rightarrow \infty$. However, we see that as $R$ approaches the finite positive value, $T$ also approaches a finite value. This suggests that we can extend our geodesic in the direction of decreasing $T$ and continue the curve in the retraction projection. As one would expect, this can be done and the extension can be constructed by considering the part of the geodesic in the RNdS coordinates which begins at $r_{bh-}$, decreases to the minimum value of $r$ and then increases to $r_{bh-}$ again. Of course, this curve lies completely within the Cauchy horizon and so, it is an unphysical extension but we can discuss it mathematically nonetheless.
\\(\textbf{Note:} For this construction to provide the correct extension of the null geodesic, we must set the static time $t$ to be decreasing on this interval).
\\Hence, we can plot the full retraction projection of this null geodesic in the Kastor-Traschen coordinates. We include a plot of the time $T$ here as a function of $\phi$ to demonstrate that it is, indeed, the projection of the whole geodesic - here the null geodesic begins at the spatial origin $R=0$ at $T \rightarrow -\infty$, traces the curve in the direction of increasing $\phi$ and returns to the origin as $T \rightarrow +\infty$:
\begin{figure}[h]
\begin{minipage}[b]{0.5\linewidth}
\centering \includegraphics[width=4.5cm,height=4.5cm]{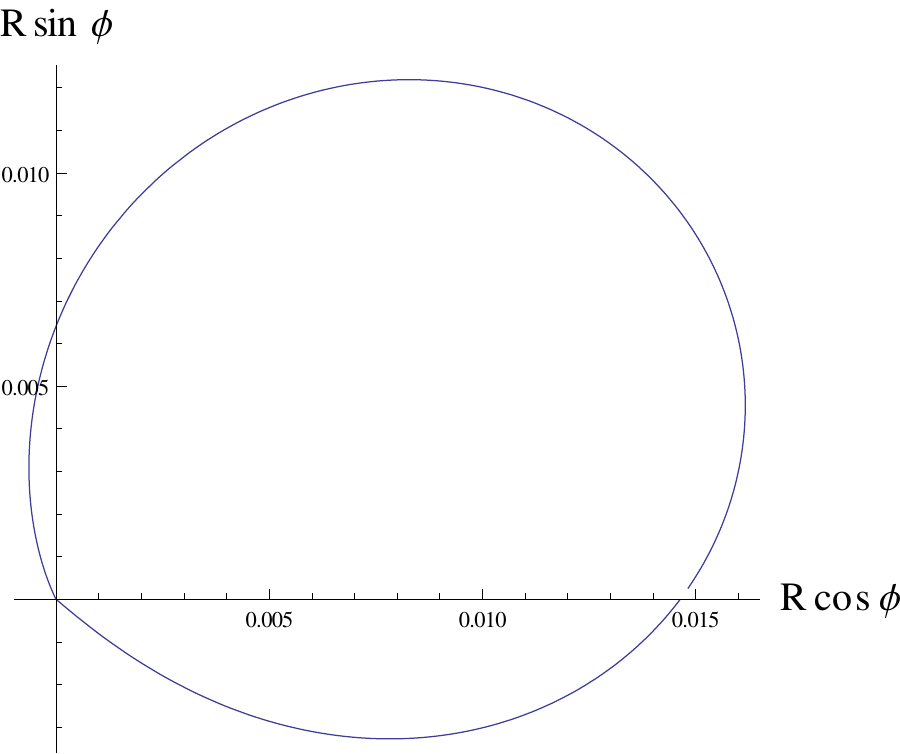}
\caption{Retraction projection of null geodesic in K-T coordinates.}
\end{minipage}
\hspace{0.5cm}
\begin{minipage}[b]{0.5\linewidth}
\centering \includegraphics[width=4.5cm, height=4.5cm]{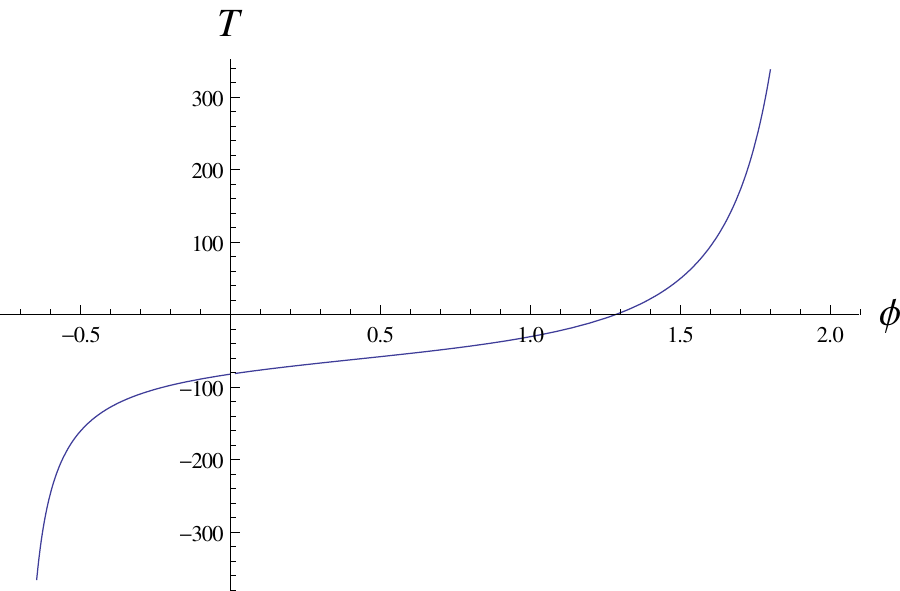}
\caption{Time $T$ as a function of $\phi$ on this interval.}
\end{minipage}
\end{figure}

\subsubsection*{Conformal Diagram}
To gain a better understanding of the geometry here, let us plot the trajectory of this null geodesic in the Penrose-Carter diagram of the spacetime. The authors in \cite{ackaymatzner} have constructed this diagram and have highlighted the region covered by the ``cosmological'' (Kastor-Traschen) coordinates. We give a copy of this diagram and include a null geodesic which runs from a point at $r<r_{bh-}$ to the outer black hole horizon $r_{bh+}$.
\begin{figure}[h]
 \begin{center}
 \label{pencart}
 \includegraphics[width=6cm,height = 6cm, angle=0]{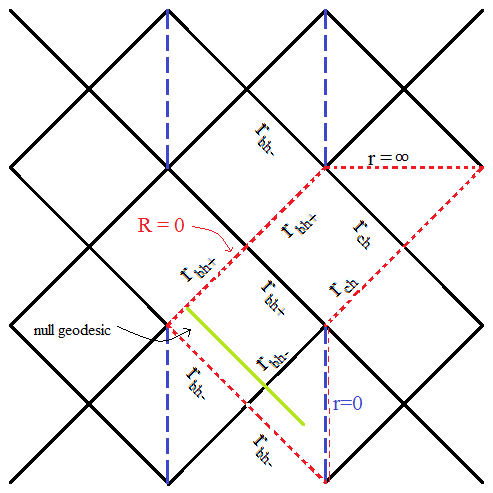}
 \caption{Penrose-Carter diagram of extremal Reissner-Nordstrom deSitter spacetime. The region bounded by the red dotted line represents a single $(T,R)$ chart.}
 \end{center}
\end{figure}
\\The authors of \cite{ackaymatzner} make the point that the cosmological coordinates ``smoothly cover the entire region from $r=0$ to $r = \infty$.'' From the diagram, we see that a single $(T,R)$ chart (bounded by the dotted red line) covers four $(t,r)$ charts encompassing all three horizons. As the light ray crosses the unstable Cauchy horizon, the static time $t \rightarrow \infty$. However in the Kastor-Traschen coordinates, we only have $T \rightarrow \infty$ as $R \rightarrow 0$ and so, mathematically, the Kastor-Traschen metric does not see a time singularity at the inner Cauchy horizon.

\section{Two-Centre Case}
Here, we make some remarks about the Kastor-Traschen metric with two black holes ($N=2$). Analytic expressions for null geodesics are much more difficult to obtain but we can make some remarks about the general theory. First of all, let us write our Kastor-Traschen metric $g$ in \emph{cylindrical polar} coordinates.
\begin{equation}
- \frac{1}{(V+cT)^2} dT^2 + (V+cT)^2 (d \rho^2 + \rho^2 d\phi^2 + dz^2)
\end{equation}
so that, without loss of generality, the singularities are placed on the $z$-axis, equidistant from the origin. That is
\begin{equation*}
V = \frac{m_1}{(\rho^2+(z-w)^2)^{1/2}} + \frac{m_2}{(\rho^2+(z+w)^2)^{1/2}}
\end{equation*}
where $m_1$ and $m_2$ are the black hole masses and $2w$ is the distance between the centres. For this metric, we have the following result regarding the retraction projection of null geodesics.
\begin{twocent}
If the retraction projection of a null geodesic has the property that its initial position and velocity lie in a plane passing through the two centres, then the entire projected null geodesic lies in that plane.
\end{twocent}
\begin{proof}
Planes which pass through both centres are characterised by the condition $\phi = constant$. Then, the retraction projection of a null geodesic lies in such a plane if and only if $\dot{\phi} = 0$ at all points on the curve. Therefore, we can verify the proposition by showing that, at any point where $\dot{\phi} = 0$, we have $\ddot{\phi} = 0$. But, from equation (\ref{twochelp}), null geodesics of Kastor-Traschen satisfy
\begin{equation*}
\ddot{\phi} + \frac{2}{\rho} \dot{\rho} \dot{\phi} = - \frac{2}{V+cT} \left(\frac{\partial V}{\partial \rho} \dot{\rho} + \frac{\partial V}{\partial z} \dot{z} \right) \dot{\phi}
\end{equation*}
and the result follows.
\end{proof}
Now let us impose the additional condition $m_1 = m_2 = M$. Then, we discover another fixed plane of null geodesics.
\begin{twocenttwo}
If the retraction projection of a null geodesic has the property that its initial position and velocity lie in the plane passing orthogonally through the midpoint of the line segment joining the two centres, then the entire projected null geodesic lies in that plane.
\end{twocenttwo}
\begin{proof}
Clearly, the plane in question is given, in cylindrical polar coordinates, by $z = 0$ and any null geodesic which lies completely in this plane will satisfy the condition $\dot{z} = 0$ at all points on the curve. Therefore, if the initial conditions $z=0$, $\dot{z} = 0$ imply that $\ddot{z} = 0$ initially, then the proposition is proved. Again, by equation (\ref{twochelp}), null geodesics of $g$ satisfy
\begin{equation*}
\ddot{z} = \frac{2}{V+cT} \left(\frac{\partial V}{\partial z} - \left(\frac{\partial V}{\partial \rho} \dot{\rho} + \frac{\partial V}{\partial z} \dot{z} \right) \dot{z} \right)
\end{equation*}
and the result follows because
\begin{equation*}
\left. \frac{\partial V}{\partial z} \right|_{z=0} = \frac{Ma}{(\rho^2 + w^2)^{3/2}} - \frac{Ma}{(\rho^2 + w^2)^{3/2}} = 0.
\end{equation*}
\end{proof}

\subsection{Third Order System describing Null geodesics}
In section 5, we discovered that the third order system used to describe the retraction projection of null geodesics of the Kastor-Traschen metric was not uniquely defined and by considering a new formulation we could obtain an interpretation of the entire set of integral curves. Here, let us take yet another system of third order ODEs, the integral curves of which contain the projected null geodesics of the one-centre Kastor-Traschen metric by eliminating the $x^i$ term from the system (\ref{mainequation}), using (\ref{replacer}) and (\ref{niceprop}) to obtain
\begin{equation}
\dddot{x}^i = -|\ddot{\textbf{x}}|^2 \dot{x}^i + \left[ \frac{2mc}{|\textbf{x}|^3} \left( \frac{\textbf{x}.\ddot{\textbf{x}}}{|\ddot{\textbf{x}}|^2} \right) - 3(\textbf{x}.\dot{\textbf{x}}) \left( \frac{1}{|\textbf{x}|^2} + \frac{\textbf{x}.\ddot{\textbf{x}}}{2(|\textbf{x}|^2 - (\textbf{x}.\dot{\textbf{x}})^2)} \right) \right] \ddot{x}^i.
\label{newode}
\end{equation}
Then, we have the following result
\begin{charac}
Any integral curve of the system of ODEs (\ref{newode}), lies in a plane. Furthermore, if $c \neq 0$, such an integral curve will coincide with a projected null geodesic of the Kastor-Traschen metric $g$ if and only if this plane passes through the origin.
\end{charac}
\begin{proof}
We construct the Frenet-Serret frame for a given integral curve of (\ref{newode})
\begin{equation*}
\textbf{T} = \dot{\textbf{x}} \,\,\,\,,\,\,\,\, \textbf{N} = \frac{\ddot{\textbf{x}}}{|\ddot{\textbf{x}}|} \,\,\,\,,\,\,\,\, \textbf{B} = \textbf{T} \times \textbf{N}.
\end{equation*}
Then, the Frenet-Serret formulas give
\begin{eqnarray*}
\dot{\textbf{T}} &=& \kappa \textbf{N}, \\
\dot{\textbf{N}} &=& -\kappa \textbf{T} + \tau \textbf{B},
\end{eqnarray*}
where $\kappa$ and $\tau$ are the curvature and torsion of the curve, respectively. The first of these equations gives us $\kappa = |\ddot{\textbf{x}}|$. Then, if we rewrite our system (\ref{newode}) in this frame, we obtain
\begin{equation*}
\dot{\textbf{N}} = -\kappa \textbf{T}
\end{equation*}
and $\tau = 0$, necessarily. Hence, a given integral curve of this system must lie in a plane.
\\We have already established that the projected null geodesics in the one-centre case will lie in a plane passing through the origin. To prove the ``if'' part of the proposition, we note that the initial data of an integral curve of (\ref{newode}) which lies on a plane through the origin is specified by six parameters - three for initial position, two for initial velocity (since it is unit in the arclength parametrisation) and one for the acceleration (in the plane of the position and velocity vectors, perpendicular to the velocity). Then, this curve is a projected null geodesic of the Kastor-Traschen metric if there exists a null geodesic with the same initial data. But, we can see that this is the case by analysing equation (\ref{replacer}). Clearly, we can specify initial position and unit velocity vectors as we please. Then, the initial acceleration vector lies in the same plane perpendicular to the velocity and we can specify its magnitude by choosing the appropriate value of $T$.
\end{proof}
So, each integral curve of (\ref{newode}) lies in some plane but unlike the case of (\ref{bigsys}), these integral curves have no obvious interpretation in terms of the Kastor-Traschen metric unless this plane passes through the origin. However, as we observe, they do arise as the projections of null geodesics for the two-centre case - in particular, those outlined by Proposition 7.2.
\\To see this let us consider the original expression for the Kastor-Traschen metric (\ref{firsteq}) with $h$ as the flat metric in Cartesian coordinates
with the potential written as follows
\begin{equation*}
V = \frac{m_1}{|\textbf{x} - \textbf{w}|} + \frac{m_2}{|\textbf{x} + \textbf{w}|}
\end{equation*}
where $\textbf{w} = (w^i)$ is a fixed vector.  For ease of notation, let us define the vectors
\begin{equation*}
X_1^i = \frac{m_1}{|\textbf{x} - \textbf{w}|^3}(x^i - w^i) \,\,\,\,\,,\,\,\,\,\, X_2^i = \frac{m_2}{|\textbf{x} + \textbf{w}|^3}(x^i + w^i).
\end{equation*}
\\Then, null geodesics of the Kastor-Traschen metric are integral curves of the following system of third order ODEs
\begin{eqnarray}
\dddot{x}^i &=& -|\ddot{\textbf{x}}|^2 \dot{x}^i - \frac{3(\textbf{X}_1+\textbf{X}_2).\ddot{\textbf{x}}(\textbf{X}_1 +\textbf{X}_2).\dot{\textbf{x}}}{2\left(\textbf{X}_1 + \textbf{X}_2 \right).\left(\textbf{X}_1 + \textbf{X}_2 - \dot{\textbf{x}}\right)} \ddot{x}^i + 2c \big(X_1^i+X_2^i- ((\textbf{X}_1 + \textbf{X}_2).\dot{\textbf{x}})\dot{x}^i \big) \nonumber \\
&-& \frac{(\textbf{X}_1+\textbf{X}_2).\ddot{\textbf{x}}}{\left(\textbf{X}_1 + \textbf{X}_2 \right).\left(\textbf{X}_1 + \textbf{X}_2 - \dot{\textbf{x}}\right)} \Biggl( \frac{3|\textbf{x} - \textbf{w}|}{m_1}(\textbf{X}_1.\dot{\textbf{x}}) \big(X_1^i - (\textbf{X}_1.\dot{\textbf{x}})\dot{x}^i \big) + \frac{3|\textbf{x} + \textbf{w}|}{m_2}(\textbf{X}_2.\dot{\textbf{x}})\big(X_2^i - (\textbf{X}_2.\dot{\textbf{x}})\dot{x}^i \big) \Biggr). \nonumber
\end{eqnarray}
This system is difficult to analyze, in general, but now let us assume that $m_1 = m_2 = M$ and restrict attention to null geodesics which lie on the plane passing through the origin, orthogonal to the line between the two centres.
\\Then $|\textbf{x} + \textbf{w}| = |\textbf{x} - \textbf{w}|$ and since the acceleration and velocity vectors are perpendicular to $\textbf{w}$, we have $\textbf{w}.\dot{\textbf{x}} = \textbf{w}.\ddot{\textbf{x}} = 0$. By making some simplifications using the geodesic equation (\ref{propnullgeod}), as in the one-centre case, we can replace the system above by
\begin{equation}
\dddot{x}^i = -|\ddot{\textbf{x}}|^2 \dot{x}^i + \left[ \frac{4Mc}{|\textbf{x} - \textbf{w}|^3} \left(\frac{\textbf{x}.\ddot{\textbf{x}}}{|\ddot{\textbf{x}}|^2} \right) - 3(\textbf{x}.\dot{\textbf{x}}) \left( \frac{1}{|\textbf{x} - \textbf{w}|^2} + \frac{\textbf{x}.\ddot{\textbf{x}}}{2(|\textbf{x} - \textbf{w}|^2 - (\textbf{x}.\dot{\textbf{x}})^2)} \right) \right] \ddot{x}^i.
\label{definer}
\end{equation}
Now we notice that this is precisely the system of third order ODEs (\ref{newode}) for the single centre case with black hole mass $m = 2M$, where $|\textbf{x} - \textbf{w}|$ represents the distance from the centre.
\\Hence, we have proved the following proposition:
\begin{finalprop}
Every null geodesic of the two-centre Kastor-Traschen metric which lies completely in the plane passing orthogonally through the midpoint of the line segment joining the centres coincides with an integral curve of the system (\ref{newode}), with mass $m = 2M$, which lies in the plane a distance $w$ from the origin, where $2w$ is the distance between the centres.
\end{finalprop}
\textbf{Remark:} This means that every integral curve of (\ref{newode}) can be realised as the retraction projection of a null geodesic in either the one-centre or two-centre Kastor-Traschen metric, making all solutions physically relevant.

\subsection{Perturbation Analysis}
Now let us look at the stability of null geodesics in the $z=0$ plane by applying a small perturbation $\epsilon \ll a$ in the $z$-direction about the origin so that $\dot{\epsilon}$ and $\ddot{\epsilon}$ are also small. Then substituting $x^3 = z + \epsilon$ into (\ref{definer}) gives us the differential equation
\begin{equation*}
\dddot{\epsilon} = - \left. |\ddot{\textbf{x}}|^2 \right|_{(z,\dot{z},\ddot{z})=0} \dot{\epsilon} + \left. \left[ \frac{4Mc}{|\textbf{x} - \textbf{w}|^3} \left(\frac{\textbf{x}.\ddot{\textbf{x}}}{|\ddot{\textbf{x}}|^2} \right) - 3(\textbf{x}.\dot{\textbf{x}}) \left( \frac{1}{|\textbf{x} - \textbf{w}|^2} + \frac{\textbf{x}.\ddot{\textbf{x}}}{2(|\textbf{x} - \textbf{w}|^2 - (\textbf{x}.\dot{\textbf{x}})^2)} \right) \right] \right|_{(z,\dot{z},\ddot{z})=0} \ddot{\epsilon}.
\end{equation*}
We can rewrite this as a coupled system of differential equations by choosing $\eta = \dot{\epsilon}$ and $\mu = \ddot{\epsilon}$ so that
\begin{equation*}
\frac{d}{ds}
\left(
\begin{array}{c}
\mu \\
\eta \end{array} \right)
= \left. \left(
\begin{array}{cc}
\frac{4Mc}{|\textbf{x} - \textbf{w}|^3} \left(\frac{\textbf{x}.\ddot{\textbf{x}}}{|\ddot{\textbf{x}}|^2} \right) - 3(\textbf{x}.\dot{\textbf{x}}) \left( \frac{1}{|\textbf{x} - \textbf{w}|^2} + \frac{\textbf{x}.\ddot{\textbf{x}}}{2(|\textbf{x} - \textbf{w}|^2 - (\textbf{x}.\dot{\textbf{x}})^2)} \right) & -|\ddot{\textbf{x}}|^2 \\
 1 & 0 \end{array} \right) \right|_{(z,\dot{z},\ddot{z})=0}
 \left( \begin{array}{c}
 \mu \\
 \eta \end{array} \right)
  \equiv B \left( \begin{array}{c}
 \mu \\
 \eta \end{array} \right).
\end{equation*}
The stability of the system under small perturbations is determined by the eigenvalues of $B$. Given that the determinant of $B$ is positive, ($=|\ddot{\textbf{x}}|^2$) we know that both eigenvalues have the same sign. If they are both positive then the system is unstable and if they are both negative then the system is stable. The mutual sign can be obtained from the trace of $B$ and thus, we get the following result,
\\\\\textbf{Proposition 7.5:} \emph{For the two-centre equal mass Kastor-Traschen metric, any geodesic which lies in the plane passing orthogonally through the midpoint of the line segment joining the two centres is stable under small perturbations normal to the plane at a point with given initial position, velocity and acceleration data if and only if}
\begin{equation*}
\frac{4Mc}{|\textbf{x} - \textbf{w}|^3} \left(\frac{\textbf{x}.\ddot{\textbf{x}}}{|\ddot{\textbf{x}}|^2} \right) - 3(\textbf{x}.\dot{\textbf{x}}) \left( \frac{1}{|\textbf{x} - \textbf{w}|^2} + \frac{\textbf{x}.\ddot{\textbf{x}}}{2(|\textbf{x} - \textbf{w}|^2 - (\textbf{x}.\dot{\textbf{x}})^2)} \right) < 0
\end{equation*}
\emph{at that point. Otherwise, it is unstable.}

\section{Unparametrised Projection of Null Geodesics in the One-Centre Kastor-Traschen Solution}
As final note, we will rewrite the system of third-order ODEs (\ref{newode}) in unparametrised form. By doing so, we get a purer notion of the set of projected null geodesics (free of parametrisation) and can make contact with the work in \cite{medvedev} where the author has explicitly derived differential invariants for systems of third order ODEs.
\\To start with, let us relabel our coordinates $x^i = (z,x^{\beta})$ with $\beta = 2,3$. If we let ' denote differentiation with respect to $z$ then we have
\begin{equation*}
\dot{x}^{\beta} = (x^{\beta})' \dot{z} \,\,\,\,,\,\,\,\, \ddot{x}^{\beta} = (x^{\beta})'' \dot{z}^2 + (x^\beta)' \ddot{z} \,\,\,\,,
\end{equation*}
\begin{equation*}
\dddot{x}^{\beta} = (x^{\beta})''' \dot{z}^3 + 3(x^{\beta})'' \dot{z} \ddot{z} + (x^{\beta})' \dddot{z}.
\end{equation*}
From the system (\ref{newode}), we can eliminate the $\dddot{z}$ term to obtain a pair of third order expressions
\begin{equation}
(x^{\beta})''' \dot{z}^3 + 3(x^{\beta})'' \dot{z} \ddot{z} = \left[ \frac{2mc}{|\textbf{x}|^3} \left(\frac{\textbf{x}.\ddot{\textbf{x}}}{|\ddot{\textbf{x}}|^2}\right) - 3(\textbf{x}.\dot{\textbf{x}}) \left( \frac{1}{|\textbf{x}|^2} + \frac{\textbf{x}.\ddot{\textbf{x}}}{2(|\textbf{x}|^2 - (\textbf{x}.\dot{\textbf{x}})^2)} \right) \right] (x^{\beta})'' \dot{z}^2.
\label{prelim}
\end{equation}
This system can be simplified even further to eliminate the factors of $\dot{z}$ and $\ddot{z}$. First, let us use the following convention
\begin{equation*}
\textbf{x} = \left(
\begin{array}{c} z \\ x^\beta
\end{array}
\right) \,\,\,\,,\,\,\,\,
\textbf{u} = \left(
\begin{array}{c} 1 \\ (x^{\beta})'
\end{array} \right) \,\,\,\,,\,\,\,\,
\textbf{a} = \left(
\begin{array}{c} 0 \\ (x^{\beta})''
\end{array} \right)
\end{equation*}
Then, using the arc-length parametrization condition $h_{jk} \dot{x}^j \dot{x}^k = 1$, one can show that
\begin{equation}
 |\textbf{u}|^2 \dot{z}^2 = 1 \,\,\,,\,\,\, \ddot{z} = - \frac{\textbf{u.a}}{|\textbf{u}|^2} \dot{z}^2,
 \label{newconds}
\end{equation}
Using this, it can also be shown that
\begin{equation*}
|\ddot{\textbf{x}}|^2 = \left( |\textbf{a}|^2 - \frac{(\textbf{u.a})^2}{|\textbf{u}|^2} \right) \dot{z}^4 \,\,\,\,,\,\,\,\, \textbf{x}.\ddot{\textbf{x}} = \left(\textbf{x.a} - \frac{(\textbf{x.u})(\textbf{u.a})}{|\textbf{u}|^2} \right) \dot{z}^2.
\end{equation*}
Then, we can use the expressions in (\ref{newconds}) to eliminate $\ddot{z}$ terms and subsequently powers of $\dot{z}$ in (\ref{prelim}) (This can be done by multiplying each term by the appropriate power of $u^2 \dot{z}^2$ to give every term the same ``weight'' in terms of powers of $\dot{z}$). The resulting expression will give us a pair of third order ODEs whose integral curves are the unparametrized curves of the system (\ref{newode}):
\begin{equation}
(x^{\beta})''' = \left[ 3 \frac{\textbf{u.a}}{|\textbf{u}|^2} + \frac{2 mc |\textbf{u}|^3}{|\textbf{x}|^3} \left( \frac{|\textbf{u}|^2 (\textbf{x.a}) - (\textbf{x.u})(\textbf{u.a})}{|\textbf{u}|^2 |\textbf{a}|^2 - (\textbf{u.a})^2} \right) - 3(\textbf{x.u}) \left(\frac{1}{|\textbf{x}|^2} + \frac{|\textbf{u}|^2 (\textbf{x.a}) - (\textbf{x.u})(\textbf{u.a})}{2(|\textbf{u}|^4|\textbf{x}|^2 - |\textbf{u}|^2(\textbf{x.u})^2)} \right) \right] (x^{\beta})''.
\label{ultimu}
\end{equation}
In the limit $c \rightarrow 0$, the second and third terms on the right-hand side of (\ref{ultimu}) vanish coinciding with the unparametrised equation for conformal circles. I have calculated some of the Medvedev invariants \cite{medvedev} can for the system (\ref{ultimu}) and for conformal circles in Mathematica and the formulae can be obtained on request (but are too long to be included with this work).

\section{Conclusions}
As was demonstrated in \cite{wernerwar}, physical phenomena observed in spacetimes which admit a timelike vector field with a certain mathematical property can often be better understood by looking at the projection of null geodesics to the space of orbits of this vector field. This structure has been well documented in the static and stationary cases and in this paper, we discussed a specific example of a metric admitting a timelike conformal retraction which is also a solution of the Einstein-Maxwell equations and so, an important GR example.
\\The third order system (\ref{thirdord}) arises naturally to describe the retraction projection of null geodesics of the Kastor-Traschen metric. We have, however, demonstrated that there is a freedom in the definition of this ODE system and that a more useful analysis is obtained by considering instead the system (\ref{bigsys}) where the totality of integral curves can be interpreted as the projection of null curves in the Kastor-Traschen metric describing a magnetic flow in the background magnetic field. This endows a physical relevance to this system and it would be interesting to probe its relevance further. Using this formulation, we've characterised those integral curves of (\ref{bigsys}) which coincide with the retraction projection of null geodesics.
\\For the one-centre K-T solution, the projected null geodesics are identified as those which lie on a plane through the origin. However, in this case, there is another third order system (\ref{newode}) whose integral curves arise as a distinguished subset of the projected null geodesics of the two-centre metric for some value of the distance between the centres, $2w$, with masses $m_1 = m_2 = M$. This analysis of null geodesics appears to be a step further than has been seen thus far but extracting more analytic solutions is far from easy.
\\There is a consistent physical intuition here if we consider what happens when $w \rightarrow 0$. We should expect to reproduce the retraction projection of some subset of the null geodesics for the one-centre Kastor-Traschen metric with black hole mass $m = 2M$ and this is precisely what happens. In fact, we obtain all of the projected null geodesics because of the inherent spherical symmetry that accompanies this limit.
\\One final point that we should stress here is that some of the physical properties of null geodesics obvious in the projection along one type of vector field can be obscured in the projection along another. A clear example of this point can be seen in the one-centre Kastor-Traschen metric which we know, via the transformation to extremal RNdS coordinates, admits both a timelike conformal retraction and a timelike static Killing vector field. Light rays project down to unparametrised geodesics of the optical metric associated to the static Killing vector field and it can be shown that for different values of the cosmological constant $\Lambda$, the resulting optical metrics are projectively equivalent. One consequence of this is the fact that the differential equations governing light rays in these spacetimes are also independent of $\Lambda$. This phenomenon is not evident, however, when we consider the retraction projection of light rays in Kastor-Traschen coordinates. There is a clear $c$ (or $\Lambda$) dependence in the system of ODEs (\ref{newode}) and this even carries over to the equations governing the unparamterised curves (\ref{ultimu}).
\\Nonetheless, this is an interesting and physically relevant area of study and this paper encourages several open questions. It would be interesting to find analytic solutions to (\ref{bigsys}) in a more general case and to say something more concrete about the Kastor-Traschen metric with arbitrary $V$. Furthermore, we may ask about the properties of the null geodesic structure of an arbitrary metric which admits a timelike conformal retraction.

\newpage
\bibliographystyle{unsrt}
\bibliography{KTbib}

\begin{thebibliography}{10}

\bibitem{wernerwar}
G.W.Gibbons, C.M.Warnick, and M.C.Werner.
\newblock Light-bending in {S}chwarzschild-de-{S}itter: projective geometry of
  the optical metric.
\newblock {\em Class. Quant. Grav.}, 25(245009), 2008.

\bibitem{blackhole}
G.W.Gibbons and C.M.Warnick.
\newblock Universal {P}roperties of {N}ear-horizon {O}ptical {G}eometry.
\newblock {\em Phys.Rev.D79:06403}, 2009.

\bibitem{self}
S.~Casey, M.~Dunajski, G.W. Gibbons, and C.~Warnick.
\newblock Optical {M}etrics and {P}rojective {E}quivalence.
\newblock {\em Phys. Rev. D}, 83(084047), 2011.

\bibitem{statmet}
G.W. Gibbons, C.A.R. Herdeiro, C.M. Warnick, and M.C. Werner.
\newblock Stationary {M}etrics and {O}ptical {Z}ermelo-{R}anders-{F}insler
  {G}eometry.
\newblock {\em Phys. Rev. D}, 79(4), 2009.

\bibitem{DGST}
M.~Dunajski, J.~Gutowski, W.~Sabra, and P.~Tod.
\newblock Cosmological {E}instein-{M}axwell {I}nstantons and {E}uclidean
  {S}upersymmetry: {A}nti-self-dual {I}nstantons.
\newblock {\em Class. Quant. Grav.}, 28(2), 2011.

\bibitem{kastor}
D.~Kastor and J.~Traschen.
\newblock Cosmological {M}ulti-black-hole {S}olutions.
\newblock {\em Phys. Rev. D}, (47), 1992.

\bibitem{yanocc}
K.~Yano.
\newblock {\em The Theory of the Lie derivative and its applications}.
\newblock North Holland, Amsterdam, 1955.

\bibitem{majumdar}
S.D. Majumdar.
\newblock A {C}lass of {E}xact {S}olutions of {E}instein's {F}ield {E}quations.
\newblock {\em Phys.Rev.}, 72(5), 1947.

\bibitem{papapetrou}
A.~Papapetrou.
\newblock A {S}tatic {S}olution of the {E}quations of the {G}ravitational
  {F}ield for an arbitrary {C}harge {D}istribution.
\newblock {\em Proc. Roy. Irish Acad. (Sect. A)}, A51(191), 1947.

\bibitem{hartle}
J.B. Hartle and S.W. Hawking.
\newblock Solutions of the {E}instein-{M}axwell {E}quations with {M}any {B}lack
  {H}oles.
\newblock {\em Commun. Math. Phys.}, 26(87), 1972.

\bibitem{bailey}
T.N. Bailey and M.~Eastwood.
\newblock Conformal circles and {P}aramterizations of {C}urves in {C}onformal
  manifolds.
\newblock {\em Proc. Amer. Math. Soc.}, 108(1), 1990.

\bibitem{friedschmid}
H.~Friedrich and B.G. Schmid.
\newblock Conformal {G}eodesics in {G}eneral {R}elativity.
\newblock {\em Proc. Royal Soc. Lond.}, 414(171-195), 1987.

\bibitem{friedrich}
H.~Friedrich.
\newblock Conformal {G}eodesics on {V}acuum {S}pace-times.
\newblock {\em Commun. Math. Phys.}, 235(513-543), 2002.

\bibitem{hackmann}
E.~Hackmann.
\newblock {\em Geodesic {E}quations in {B}lack hole {S}pace-times with
  {C}osmological {C}onstant}.
\newblock PhD thesis, Universit\"{a}t Bremen, 2010.

\bibitem{stuchlik}
Z.~Stuchlik and M.~Calvani.
\newblock Null {G}eodesics in {B}lack {H}ole {M}etrics with {N}on-{Z}ero
  {C}osmological {C}onstant.
\newblock {\em Gen. Rel. and Grav.}, 23(5), 1991.

\bibitem{mellor}
F.~Mellor and I.~Moss.
\newblock A {R}eassessment of the {S}tability of the {C}auchy {H}orizon in
  de{S}itter {S}pace.
\newblock {\em Class. Quant. Grav.}, 9(L43-L46), 1992.

\bibitem{penrose}
R.~Penrose.
\newblock Gravitational {C}ollapse: the role of {G}eneral {R}elativity.
\newblock {\em Nuovo Cimento}, 1(252-276), 1969.

\bibitem{romans}
L.J. Romans.
\newblock Supersymmetric, {C}old and {L}ukewarm {B}lack {H}oles in
  {C}osmological {E}instein-{M}axwell {T}heory.
\newblock {\em Nuclear Phys. B}, 383(395-415), 1992.

\bibitem{ackaymatzner}
S.~Akcay and R.A. Matzner.
\newblock Kerr-de{S}itter {U}niverse.
\newblock {\em Class. Quant. Grav.}, 28(085012), 2011.

\bibitem{medvedev}
A.~Medvedev.
\newblock Third {O}rder {O}{D}{E}s {S}ystems and its {C}haracteristic
  {C}onnections.
\newblock {\em arXiv:1104.0965}, 2011.

\end{thebibliography}

\end{document}